%% file: main.tex
\documentclass{article}

\usepackage{arxiv}

\usepackage[utf8]{inputenc} % allow utf-8 input
\usepackage[T1]{fontenc}    % use 8-bit T1 fonts
\usepackage{hyperref}       % hyperlinks
\usepackage{url}            % simple URL typesetting
\usepackage{amsfonts}       % blackboard math symbols
\usepackage{nicefrac}       % compact symbols for 1/2, etc.
\usepackage{microtype}      % microtypography
\usepackage{graphicx}
\usepackage[numbers]{natbib}
\usepackage{doi}

\usepackage{amsthm}
\usepackage{algpseudocode}
\usepackage{algorithm}
\usepackage{tikz}
\usepackage{multirow}
\usepackage{amsmath}
\usepackage{array}
\usepackage{hyperref}
\usepackage{tabularx}
\usepackage{nicematrix}
\usepackage{subcaption}
\usepackage{adjustbox}
\usepackage{makecell}
\usepackage{todonotes}
\usepackage{bm}
\usepackage{arydshln}
\usepackage{booktabs}

\floatname{algorithm}{Mechanism}

\DeclareMathOperator*{\argmax}{arg\,max}
\DeclareMathOperator*{\argmin}{arg\,min}
\def \t {\mathbf{t}}
\def \optt {a^*}
\def \z {z}
\def \f {f}
\def \err {\hat{\rho}}
\def \pred {\hat{a}}
\def \C {C}
\def \W {W}
\def \p {p}
\def \b {\mathbf{b}}
\newcommand{\mech}{\text{\sc Mech}}
\newcommand{\opt}{\text{\sc Opt}}

\newtheorem{remark}{Remark}
\newtheorem{theorem}{Theorem}
\newtheorem{lemma}{Lemma}
\newtheorem{claim}{Claim}

\title{Mechanism design augmented with output advice}

\author{ \href{https://orcid.org/0000-0002-9623-7461}
    {\includegraphics[scale=0.06]{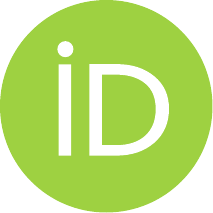}\hspace{1mm}George Christodoulou} \\
	Aristotle University of Thessaloniki\\
	Archimedes/RC Athena\\
	Greece \\
	\AND
	\href{https://orcid.org/0000-0003-3997-5131}{\includegraphics[scale=0.06]{orcid.pdf}\hspace{1mm}Alkmini Sgouritsa} \\
	Athens University of Economics and Business\\
	Archimedes/RC Athena\\
	Greece \\
    \AND
	{\hspace{1mm}Ioannis Vlachos} \\
	Athens University of Economics and Business\\
	Archimedes/RC Athena\\
	Greece \\
}

\renewcommand{\headeright}{}
\renewcommand{\undertitle}{}
\renewcommand{\shorttitle}{Mechanism design augmented with output advice}

\hypersetup{
pdftitle={Mechanism design augmented with output advice},
pdfsubject={mechanism design},
pdfauthor={I. Vlachos},
pdfkeywords={mechanism design, output advice, facility location, scheduling, house allocation, auctions},
}

\begin{document}
\maketitle

\begin{abstract}

Our work revisits the design of mechanisms via the {\em learning-augmented} framework. In this model, the algorithm is enhanced with imperfect (machine-learned) information concerning the {\em input}, usually referred to as prediction. The goal is to design algorithms whose performance degrades gently as a function of the prediction error and, in particular, perform well if the prediction is accurate, but also provide a worst-case guarantee under any possible error. This framework has been successfully applied recently to various mechanism design settings, where in most cases the mechanism is provided with a prediction about the {\em types} of the players.

We adopt a perspective in which the mechanism is provided with an {\em output recommendation}. We make no assumptions about the quality of the suggested outcome, and the goal is to use the recommendation to design mechanisms with low approximation guarantees whenever the recommended outcome is reasonable, but at the same time to provide worst-case guarantees whenever the recommendation significantly deviates from the optimal one. We propose a generic, universal measure, which we call {\em quality of recommendation}, to evaluate mechanisms across various information settings. We demonstrate how this new metric can provide  refined analysis in existing results. 

This model introduces new challenges, as the mechanism receives limited information comparing to settings that use predictions about the types of the agents. We study, through this lens, several well-studied mechanism design paradigms, devising new mechanisms, but also providing refined analysis for existing ones, using as a metric the quality of recommendation. We complement our positive results, by exploring the limitations of known classes of strategyproof mechanisms that can be devised using output recommendation.

\end{abstract}


\newpage

\section{Introduction}
\input{intro}

\subsection{Related Work}\label{sec:related}
\input{related}

\section{Model}
\input{model}

\section{Facility Location}\label{sec:facility}
\input{facilityLocation}

\section{Scheduling}\label{sec:scheduling}
\input{scheduling}

\section{House Allocation}\label{sec:houseAllocation}
\input{houseAllocation}

\section{Auctions}\label{sec:combinatorialAuctions}
\input{combAuctions}

\bibliographystyle{plainnat}
\bibliography{bibliography}

\newpage

\appendix
\input{appendix}

\end{document}

%% file: intro.tex
Motivated by the occasionally overly pessimistic perspective of
worst-case analysis, a recent trend has emerged focusing on the design
and analysis of algorithms within the so-called {\em
  learning-augmented framework} (refer to \cite{mitzenmacher2022algorithms}
for an overview). Within this framework, algorithms are enhanced with
imperfect information about the input, usually referred to as
{\em predictions}. These predictions can stem from machine learning models,
often characterized by high accuracy, leading to exceptional performance. However, their accuracy is not guaranteed, so the
predicted input may differ significantly from the actual
input. Blindly relying on these predictions can have significant
consequences compared to employing a worst-case analysis approach.

The framework aims to integrate the advantages of both approaches. The
goal is to use these predictions to design algorithms whose
performance degrades gently as a function of the inaccuracy of the prediction, known as the prediction error. In
particular, they should perform well whenever the prediction is
accurate --a property known as {\em consistency}-- and also provide a worst-case
guarantee under any possible error --a property known as {\em robustness}.

\citet{xu2022mechanism} and \citet{agrawal2022learning} applied the learning-augmented
framework in mechanism design settings, where there is incomplete
information regarding the preferences (or types) of the participants
over a set of alternatives. Traditional mechanism design addresses
this information gap by devising strategyproof mechanisms that offer
appropriate incentives for agents to report their true types. In the
learning-augmented model, it is generally assumed that the mechanism
is equipped with predictions about the types of the agents. The aim is
to leverage these predicted types to design strategyproof mechanisms that
provide consistency and robustness guarantees. Since then, this
model has found application in diverse mechanism design settings
\cite{balkanski2022strategyproof,caragiannis2024randomized,lu2023competitive,istrate2022mechanism}.

\paragraph{Mechanisms with output advice}
In this work, we propose an alternative perspective on mechanism
design with predictions. We assume that the mechanism is provided with
external advice to \emph{output a specific outcome}, rather being provided with predictions of the agents' types. For example, in a job
scheduling problem, the designer may receive a recommended partition
of tasks for the machines, rather than a prediction about the
machines' processing times. Similarly, in an auction setting, an
allocation of goods is provided, rather than a prediction about the
agents' valuations.

Following the tradition of the learning-augmented framework, we make
no assumptions about the quality of the recommended outcome, which may
or may not be a good fit for the specific (unknown) input. The goal is
to use the recommendation to design a strategyproof mechanism with good
approximation guarantees whenever the recommended outcome is a good
fit, but at the same time provide worst-case guarantees whenever
the recommendation deviates from the optimal one.

We observe that one can reinterpret previous models within the framework of our model, viewing it as a more constrained version of
predictions with \emph{limited information}.\footnote{For example, in
  \cite{agrawal2022learning}, it is assumed that the mechanism is
  provided with the optimal allocation with respect to the predicted
  types. Refer to the discussion in Section~\ref{sec:facility} for a
  comparison and differences with their model.}Since we only require
limited information regarding the outcome, our model may be better
suited to handle cases where historical input data is absent or
limited, which may occur for various reasons such as privacy concerns,
data protection, challenges in anonymizing, or simply because the
information is missing. For instance, historical data in an auction
may sometimes only contain information about the winners and perhaps the prices, omitting details about their exact valuation or the values of those
who lost. Additionally, our model may be applied in cases where the
designer does not need to know the specifics of the algorithm and
treats it as a black box, as long as it yields satisfactory
allocations, even if the inner workings are not fully understood.

We make no assumption about \emph{how} the outcome recommendation was
produced, which makes it quite general and adaptable to different
application domains. For instance, the outcome may represent the
optimal allocation with respect to predicted data (as seen in
\cite{agrawal2022learning}), or it may be a solution generated by an
approximation algorithm or a heuristic. Consequently, the quality of
the recommended outcome may be affected by various factors, such as
the accuracy of the predicted data or the limitations of computational
resources which prevent the computation of optimal solutions, even when
the data is accurate.

A beneficial side effect of our model is that an outcome recommendation fits
in a plug-and-play fashion with a generic machinery for strategyproofness
in multi-dimensional mechanism design, particularly maximal in range
VCG mechanisms (or more generally with affine maximizers) in a straightforward manner: we simply add the recommended outcome to the range of the affine
maximizer (see Section \ref{sec:combinatorialAuctions}).

\paragraph{Quality of recommendation}
In the learning augmented framework, the performance of an algorithm (or mechanism) is evaluated based on the \emph{prediction error}, which quantifies the disparity between the predicted and actual data. Unfortunately, there is no universal definition for such an error; it is typically domain-specific (e.g., the ratio of processing times for scheduling \cite{lu2023competitive, balkanski2022strategyproof} or (normalized) geometric distance for facility location \cite{agrawal2022learning}). 
%and, in some cases, even algorithm-specific. 
Therefore, if one modifies the information data model for a specific problem—for instance, by assuming that only a fraction or a signal of the predicted data is provided—it becomes necessary to redefine the prediction error.

To address this issue, we propose a generic, universal measure that
can be applied to analyze algorithms across various information
settings and application domains. We define the \emph{quality of
  recommendation} as the approximation ratio between the cost (or welfare) of the
recommended outcome and the optimal cost (or welfare) both evaluated w.r.t the
actual input.

It is worth emphasizing that although the above definition aligns
naturally with our information model, as we do not assume the designer is provided with 
predicted data, it can also be applied to richer information models
with partial or even full predicted input.

We argue that it provides a unified metric for settings involving
predictions, particularly when the objective is to design mechanisms
(or more generally algorithms) with low approximation or competitive
ratio. The disparity between predicted and actual data, captured by the predicted error, may not always be relevant and can
lead to misleading evaluations; there are cases where this error may
be significantly large, but the optimal solution remains largely
unchanged. For example, consider the problem of makespan minimization
in job scheduling (see also Section~\ref{sec:facility} for a detailed
example in facility location). In \cite{lu2023competitive,
  balkanski2022strategyproof}, the prediction error used is the
maximum ratio of processing times, and it appears in the approximation
guarantees. There are simple instances where this ratio is
arbitrarily large, but the optimal allocation remains the
same. Consequently, when the prediction error is incorporated into the
analysis, it may lead to overly pessimistic guarantees for mechanisms
that perform much better (see Section~\ref{sec:facility}).
Our metric avoids such pathological situations.

\subsection{Contributions}
We propose studying mechanisms augmented with output advice, a setup
that utilizes limited information to provide improved approximation
guarantees. Additionally, we introduce a unified metric that can
provide more accurate evaluations, even for settings with richer information models.

We explore the limitations of the class of strategyproof mechanisms that can be devised using this limited information across various mechanism design settings.

\paragraph{Facility Location}
In the facility location problem, there are $n$ agents each with a preferred location and the goal is to design a strategyproof mechanism that determines the optimal facility location based on an objective. In Section~\ref{sec:facility}, we derive new approximation bounds for the facility location
problem revisiting the Minimum Bounding Box and the
Coordinatewise Median mechanisms defined in
~\cite{agrawal2022learning}, as a function of the quality of
recommendation. We provide tight bounds, and demonstrate that in some
cases they outperform previous analysis with the use of a prediction
error.

\paragraph{Scheduling} In Section~\ref{sec:scheduling}, we study a scheduling problem with 
unrelated machines, where each machine has a cost for each job, which corresponds to the processing time of the job on the machine. Each job is assigned to exactly one machine, and the goal is to minimize the makespan having an output allocation as a
recommendation. We devise a new strategyproof mechanism (Mechanism \ref{alg:AllocationScaledGreedy}), that takes
also as input a confidence parameter $\beta\in [1,n]$, reflecting the level of trust in the recommendation. We show that
this mechanism is $(\beta+1)$-consistent and
$\frac{n^2}{\beta}$-robust (Theorem \ref{the:all_greedy}). Altogether, we obtain a 
$\min\{(\beta+1) \err, n+\err, \frac{n^2}{\beta}\}$ upper bound on the approximation ratio, where $\err$ is the quality of the recommendation, that we show that is asymptotically tight (Theorem \ref{thm:ASGsmoothness}). We complement this positive result, by showing that, given only the
outcome as advice, it is impossible to achieve a better
consistency-robustness trade-off in the class of
the weighted VCG mechanisms (Theorem \ref{the:asg_opt}).

\paragraph{House Allocation} Next, we switch to the house allocation
problem that can be found in Section~\ref{sec:houseAllocation}. In this problem, we aim to assign $n$ houses to a set of $n$ agents in a way that ensures strategyproofness and maximizes the social welfare. We use the TTC mechanism with the recommendation as an initial endowment (Mechanism \ref{alg:ttc}), and prove that this is
$\min\{\err, n\}$-approximate for unit-range valuations and
$\min\{\err, n^2\}$-approximate for unit-sum valuations, where $\err$ is the quality of recommendation (Theorem \ref{the:ttc_smoothness_tight}). Finally, we prove it is optimal among strategyproof, neutral and nonbossy mechanisms (Theorem \ref{the:opt_ha}) using the characterization of \cite{svensson1999strategy}.

\paragraph{Combinatorial Auctions} Finally, we study combinatorial auctions given a recommended allocation (see Section~\ref{sec:combinatorialAuctions}). In the combinatorial auctions setting, there is a set of $m$ indivisible
objects to be sold to $n$ bidders, who have private values
for each possible bundle of items. We observe that our advice
model fits nicely with the maximal in range VCG mechanisms or more
generally with the affine maximizers, by preserving strategyproofness (Mechanism \ref{alg:mir}). These mechanisms provide the best
known bounds for the approximation of the maximum social welfare for several
classes of valuations~\cite{dobzinski2005approximation,DBLP:journals/jair/DobzinskiN10,DBLP:journals/geb/HolzmanKMT04}. By including the
recommended outcome in the range of the affine maximizer, we
immediately obtain $1$-consistency, while maintaining the robustness
guarantees of those mechanisms (see Lemma \ref{lem:mir}).  

\begin{table} [h]
    \caption{Contribution Results. Consistency, robustness and approximation results proved for the mechanism design problems augmented with \textit{output} advice. In the house allocation problem, bounds are shown for unit-range valuations, while the ones in parentheses are for unit-sum valuations. In auctions, $\rho_M$ is the approximation ratio guarantee of a maximal in range mechanism.}
    \centering
    \begin{tabular}{|c|c|c|c|}\hline 
         \textbf{Problem} & \textbf{Cons} &  \textbf{Rob} &  $\f(\t, \err)$\textbf{-approximation} \\ \hline 
        Facility Location (egalitarian)& 1 \cite{agrawal2022learning}&  1+$\sqrt{2}$ \cite{agrawal2022learning}&  $\min\{\err, 1+\sqrt{2}\}$ \\\hline
        Facility Location (utilitarian) &  $\frac{\sqrt{2\lambda^2+2}}{1+\lambda}$ \cite{agrawal2022learning}&  $\frac{\sqrt{2\lambda^2+2}}{1-\lambda}$ \cite{agrawal2022learning}&  $\min\{\sqrt{2}\err, \err + \sqrt{2}, \frac{\sqrt{2\lambda^2+2}}{1-\lambda}\}$\\ \hline
        Scheduling&  $\beta+1$&  $\frac{n^2}{\beta}$&  $\min\{(\beta + 1)\err, n+\err, \frac{n^2}{\beta}\}$ \\ \hline
         House Allocation&  1 &  $n$ (or $n^2$) &  $\min\{\err, n\text{ (or } n^2)\}$\\\hline
         Auctions &  1&  $\rho_M$&  $\min\{\err, \rho_M\}$\\ \hline 
    \end{tabular}
    \label{tab:contrib}
\end{table}

%% file: related.tex
\paragraph{Learning-augmented mechanism design}
Recently, there has been increased interest in leveraging predictions to improve algorithms' worst case guarantees. The influential framework of ~\citet{lykouris2021competitive} applied on caching, formally introduced the notions of consistency and robustness, under minimal assumptions on the machine learned oracle. In the survey of \citet{mitzenmacher2022algorithms}, black box originated predictions are employed to circumvent worst-case analysis on several algorithmic problems such as ski rental and learned bloom filters.

In the domain of algorithmic game theory, a multitude of challenges have been addressed through the lens of predictive methodologies. \citet{agrawal2022learning} initiated the design and analysis of strategyproof mechanisms augmented with predictions regarding the private information of the participating agents by revisiting the problem of facility location with strategic agents. \citet{xu2022mechanism} applied the algorithmic design with predictions framework on several mechanism design problems including revenue-maximizing single-item auction, frugal path auction, scheduling, and two-facility location. Another version of the facility location problem, obnoxious facility location, was studied by ~\citet{istrate2022mechanism}. ~\citet{balcan2023bicriteria} developed a new methodology for multidimensional mechanism design that uses side information with the dual objective of generating high social welfare and high revenue. Strategyproof scheduling of unrelated machines was studied in \cite{balkanski2022strategyproof}, achieving the best of both worlds using the learning-augmented framework.

Revenue maximization is also considered in \cite{balkanski2023online} in the online setting, while ~\citet{lu2023competitive} study competitive auctions with predictions. ~\citet{caragiannis2024randomized} assume that the agent valuations belong to a known interval and study single-item auctions with the objective of extracting a large fraction of the highest agent valuation as revenue. 

Other settings enhanced with predictions include the work of ~\citet{gkatzelis2022improved}, where predictions are applied to network games and the design of decentralized mechanisms in strategic settings. In \cite{berger2023optimal}, the scenario includes a set of candidates and a set of voters, and the objective is to choose a candidate with minimum social cost, given some prediction of the optimal candidate.

\paragraph{Facility Location}
For single facility location on the line, the mechanism that places the facility on the median over all the reported points is strategyproof and optimal for the utilitarian objective. Additionally, it achieves a 2-approximation for the egalitarian social cost, the best approximation achievable by any deterministic and strategyproof mechanism~\cite{procaccia2013approximate}. In the two-dimensional Euclidean space, the Coordinatewise Median mechanism achieves a $\sqrt{2}$-approximation for the utilitarian objective~\cite{DBLP:conf/sagt/Meir19}, and a 2-approximation for the egalitarian objective~\cite{goel2023optimality}. These approximation bounds are both optimal among deterministic and strategyproof mechanisms.

In \cite{agrawal2022learning}, the facility is used as the prediction to improve the above results. Concerning the egalitarian social cost and the two-dimensional version of the problem, they achieve perfect consistency, and a robustness of $1+\sqrt{2}$. They also prove that their mechanism provides an optimal trade-off between robustness and consistency. Regarding the utilitarian social cost in two dimensions, they propose a deterministic mechanism achieving $\frac{\sqrt{2\lambda^2+2}}{1+\lambda}$-consistency, $\frac{\sqrt{2\lambda^2+2}}{1-\lambda}$-robustness and optimal trade-off among deterministic, anonymous, and strategyproof mechanisms. Finally, they define a custom prediction error and provide general approximation guarantees, for both of the objectives.
\\
\paragraph{Scheduling}
~\citet{christodoulou2023proof} validated the conjecture of Nisan and Ronen, and proved that the best approximation ratio of deterministic strategyproof mechanisms for makespan minimization for $n$ unrelated machines is $n$. Even if we allow randomization, the best known approximation guarantee achievable by a randomized strategyproof mechanism is $O(n)$ \cite{christodoulou2010mechanism}.

Following the prediction framework, ~\citet{xu2022mechanism} study the problem with predictions $\hat{t}_{ij}$ denoting the predicted processing time of job $j$ by machine $i$. They propose a deterministic strategyproof mechanism with an approximation ratio of $O(\min\{\gamma \eta^2, \frac{m^3}{\gamma^2}\})$, where $\gamma \in [1,m]$ is a configurable consistency parameter and $\eta \geq 1$ is the prediction error. ~\citet{balkanski2022strategyproof} extend these results by identifying a deterministic strategyproof mechanism that guarantees a constant consistency with a robustness of $2n$, achieving the best of both worlds.

\paragraph{House Allocation}
Over the years, several different mechanisms have been proposed with various desirable properties related to strategyproofness, fairness, and economic efficiency, with Probabilistic Serial and Random Priority being the two prominent examples. In ~\cite{filos2014social}, it was proved that using randomized mechanisms, the approximation ratio of the problem is $\Theta(\sqrt{n})$, using the Random Priority Mechanism, and that this is optimal among all strategyproof mechanisms. A lower bound of $\Omega(n^2)$ on the \textit{Price of Anarchy} for any deterministic mechanism (not necessarily strategyproof) is proved in ~\cite{DBLP:conf/atal/ChristodoulouFF16a}. In \cite{DBLP:journals/jair/AmanatidisBFV22}, a $\Theta(n^2)$ bound is proved for the distortion of all \textit{ordinal} deterministic mechanisms. There exist lower bounds for all deterministic strategyproof mechanisms which are $\Omega(n^2)$ for unit-sum and $\Omega(n)$ for unit-range,
respectively. To the best of our knowledge there is no single point of reference, for these bounds, but can follow from known results in the literature, after observing that deterministic strategyproof mechanisms are ordinal, see~\cite{DBLP:conf/atal/ChristodoulouFF16a,DBLP:journals/jair/AmanatidisBFV22}.

\paragraph{Auctions}
The central positive technique of strategyproof auction design is the VCG payment scheme. Unfortunately, while VCG works effectively from a game-theoretic perspective, it requires calculating the optimal solution, which is inherently intractable. As a result, current research focuses on designing strategyproof optimal auctions constrained to polynomially many queries. It is noteworthy that the design of strategyproof, near-optimal auctions using neural networks~\cite{dutting2019optimal, shen2018automated} has been studied extensively for automated mechanism design.

Auctions incorporating predictions have been explored across various settings such as revenue maximization auctions~\cite{caragiannis2024randomized, xu2022mechanism}, competitive auctions~\cite{lu2023competitive} and the online setting~\cite{balkanski2023online}. In our work, we consider applications on important auction variations, namely multi-unit auctions~\cite{DBLP:journals/jair/DobzinskiN10}, combinatorial auctions with general bidders~\cite{DBLP:journals/geb/HolzmanKMT04} and combinatorial auctions with subadditive valuations~\cite{dobzinski2005approximation}.

%% file: model.tex
We consider various mechanism design scenarios that fall into the
following abstract mechanism design setting. There is a set of $n$
agents and a (possibly infinite) set of alternatives $\mathcal{A}$.
Each agent (player) $i\in \{1,\ldots, n\}$ can express their
preference over the set of alternatives via a valuation function
$t_{i}$ which is private information known only to them (also called
the {\em type} of agent $i$). The set $\mathcal{T}_i$ of possible
types of agent $i$ consists of all functions
$b_i:\mathcal{A} \rightarrow \mathbb{R}$. Let also
$\mathcal{T} = \times_{i\in N}\mathcal{T}_i$ denote the space of type
profiles.

A mechanism defines for each player $i$ a set $\mathcal{B}_i$ of
available strategies the player can choose from. We consider
\emph{direct revelation} mechanisms, i.e.,
$\mathcal{B}_i=\mathcal{T}_i$ for all $i,$ meaning that the players'
strategies are to simply report their types to the mechanism. Each
player $i$ provides a \emph{bid} $b_i\in \mathcal{T}_i$, which may not
necessarily match their true type $t_i$, if this serves their
interests. A mechanism $(f,\p)$ consists of two parts:
\paragraph{A selection algorithm:} The selection algorithm $f$
selects an alternative  based on the players' inputs
(bid vector) $b=(b_1,\ldots ,b_n)$. We denote by $f(\mathbf{b})$ the alternative chosen for the bid vector $\b = (b_1,\ldots,b_n)$.
\paragraph{A payment scheme:} The payment scheme
$\p=(p_1,\ldots, p_n)$ determines the payments, which also depend on
the bid vector $\b$. The functions $\p_1,\ldots,\p_n$ represent the
payments that the mechanism hands to each agent, i.e.,
$\p_i:\mathcal{T}\rightarrow \mathbb{R}$.

The {\em utility} $u_i$ of an agent $i$ is the {\em actual} value they gain from the chosen alternative minus the payment they will have to pay, $u_i(\b)=t_i(f(\b))-p_i(\b)$. We consider  \emph{strategyproof} mechanisms. A mechanism is strategyproof, if for every agent, reporting their true type is a \emph{dominant strategy}. Formally,
$$u_i(t_i,\b_{-i})\geq u_i(t'_i,\b_{-i}),\qquad \forall i\in [n],\;\;
t_i,t'_i\in \mathcal{T}_i, \;\; \b_{-i}\in \mathcal{T}_{-i},$$ where $\mathcal{T}_{-i}$ denotes all parts of $\mathcal{T}$ except its $i$-th part.

In some of our applications (e.g. in the facility location and
scheduling settings), it is more natural to consider that the agents
are cost-minimizers rather than utility-maximizers. Therefore, for convenience
we will assume that each agent $i$ aims to minimize a cost
function rather than maximizing a utility function. We stress that some of our applications (e.g. facility location, one-sided matching)
fall into mechanism design without money,
hence in those cases we will assume $\p_i(\t)=0, \forall \t$ and
$i \in [n]$.

\paragraph{Social objective} We assume that there is an underlying
objective function that needs to be optimized. We consider both {\em cost minimization} social objectives~(facility location in Section~\ref{sec:facility}, scheduling in Section~\ref{sec:scheduling}, and {\em welfare maximization} (house allocation in Section~\ref{sec:houseAllocation}, auctions in Section~\ref{sec:combinatorialAuctions}). In the context of a maximization
problem, we assume that we are given a welfare function
$\W: \mathcal{T}\times \mathcal{A}\rightarrow \mathbb{R}_+$. If all
agents' types were known, then the goal would be to select the
outcome $a$ that maximizes $\C(\t,a)$.

The quality of a mechanism for a given type vector $\t$ is measured by
the welfare $\mech(\t)$ achieved by its selection algorithm $f$,
$ \mech(\t) = \W(\t,f(\t))$, which is compared to the optimal cost
$\opt(\t)=\max_{a\in \mathcal{A}}\W(\t,a)$. We denote an optimal alternative for a given bid vector $\t$ by $\optt$.

In most application domains, it is well known that only a subset of
algorithms can be selection algorithms of strategyproof mechanisms. In
particular, no mechanism's selection algorithm is optimal for every
$t$, prompting a natural focus on the approximation ratio of
the mechanism's selection algorithm. A mechanism is
\emph{$\rho$-approximate}, for some $\rho\geq 1$, if its selection algorithm is
$\rho$-approximate, that is, if $\rho\geq\frac{\opt(\t)}{\mech(\t)}\;$ for all
possible inputs $\t$.

\paragraph{Mechanisms with advice}
We assume that in addition to the input bid $\b$, the mechanism is
also given as a recommendation/advice, a predicted alternative
$\pred \in \mathcal{A}$, but without any guarantee of its quality.
A natural requirement, known as  {\em consistency}, requires that whenever the recommendation is
accurate, then the mechanism should achieve low approximation. A mechanism is said to be $\beta$-consistent if it
is $\beta$-approximate when the prediction is accurate, that is, the
predicted outcome $\pred$ is optimal for the given $\t$ vector.  On
the other hand, if the prediction is poor,  {\em
  robustness} requires that the mechanism retains some reasonable worst-case guarantee. A mechanism is said to be $\gamma$-robust if it is
$\gamma$-approximate for all predictions:

$$\max_{\t}{\frac{\opt(\t)}{\mech(\t, \optt)} \leq \beta} \,;\qquad\qquad \max_{\t, \pred}{\frac{\opt(\t)}{\mech(\t, \pred)}} \leq \gamma\,.$$
In order to measure the quality of the prediction, we define the {\em recommendation error}, denoted by $\err$, as the approximation ratio of the cost of the recommended outcome to the optimal cost i.e., $$\err = \frac{\opt(\t)}{\W(\t,\pred)}.$$

In some of our applications, the social objective is a cost minimization problem, where there is an underlying social cost function
$\C: \mathcal{T}\times \mathcal{A}\rightarrow \mathbb{R}_+$ that needs
to be minimized. We adapt our definitions for approximation and for the prediction error accordingly.

In particular, the quality of a mechanism for a given type vector $\t$ is measured by
the cost $ \mech(\t, \pred) = \C(\t,f(\t, \pred))$, which is compared to the optimal social cost
$\opt(\t)=\min_{a\in \mathcal{A}}\C(\t,a)$. A mechanism is $\rho$-approximate, if $\rho\geq\frac{\mech(t)}{\opt(t)}\;$ for all
possible inputs $t$. Similarly to the maximization version, $\beta$-consistency and $\gamma$-robustness are defined as:

$$\max_{\t}{\frac{\mech(\t, \optt)}{\opt(\t)} \leq \beta} \,;\qquad\qquad \max_{\t, \pred}{\frac{\mech(\t, \pred)}{\opt(\t)}} \leq \gamma\,$$

while the recommendation error is defined as the approximation ratio 
\[\err = \frac{\C(\t,\pred)}{\opt(\t)}.\]

Note that for both versions, the quality of recommendation $\err$ exceeds $1$, with $1$ indicating perfect quality and higher values indicating poorer quality. 

Additionally, we require a smooth decay of the approximation ratio as a function of the quality of the recommendation as it moves from being perfect to being arbitrarily bad. We say that an algorithm is {\em smooth} if its approximation ratio degrades
at a rate that is at most linear in $\err$~\cite{antoniadis2023online,antoniadis2023online1,purohit2018improving}.

%% file: facilityLocation.tex
In this section, we study mechanisms for the facility location problem in
the two-dimensional Euclidean space. There are $n$ agents each with a preferred (private) location 
 $\z_i=(x_i,y_i), 1\leq i \leq n$ in $\mathbb{R}^2$. The goal of the mechanism is to aggregate the preferences of the agents and determine the optimal facility location at a point $f(\t)$ in $\mathbb{R}^2$. Given a facility at point $a\in \mathbb{R}^2$, the private cost $t_i(a)$ of each agent is measured by the distance of $\z_i$ from $a$, i.e., $t_i(a)=d(\z_i, a)$, and the private objective of each agent is to minimize their cost. Two different social cost functions have been used to evaluate the quality of a location $a$~\cite{agrawal2022learning}; the {\em egalitarian cost}, which measures the maximum cost incurred by $a$ among all
agents $\C(\t,a) = \max_i{t_i(a)}$, and the {\em utilitarian} cost, which considers the sum of the individual costs i.e., $C(\t,a)=\sum_{i}t_i(a)$.

We assume that the mechanism is equipped with a recommended point $\pred\in \mathbb{R}^2$. This is perceived as a recommendation to place the facility at $\pred$. For a given $\t$ we denote by $\optt(\t)$ the optimal location minimizing the social cost, and by $\err(\t)$ the quality of the recommended outcome, which is defined as the approximation ratio $C(\t,\pred)/\opt(\t)$ and measures the approximation that would by achieved by placing the facility at $\pred$.
We use the simpler notation $\optt$ and $\err$ when $\t$ is clear from the context.

We note that for this problem our model coincides with the model studied in \cite{agrawal2022learning}
for facility location problems, although our perspective is slightly different.
Their paper considers that the missing information is the type of the agents, and they assume that they receive a {\em signal of the predicted input} $\pred$, the optimal location w.r.t. the predicted types. Due to this perspective, they defined as {\em prediction error} the (normalized) distance of their prediction, comparing to the optimal solution w.r.t the actual types. We perceive $\pred$ as an output advice. Clearly, one can interpret the output as a signal of some sort of predicted data. However, we treat the advice as a recommendation, with unknown quality, and under this perspective in the context of this paper, it makes more sense to measure it by the approximation ratio w.r.t the actual (but unknown) input.   

We showcase this effect in the following example of the facility location
problem in the line for the utilitarian social cost, and we further discuss it in Section~\ref{sec: error}.  Consider $2m-1$ agents, see Figure
\ref{fig:near_optimal}, whose preferred locations are clustered in two different
points, the one at position $(0,0)$ and the other at position $(1,0)$, where the first point is preferred by $m$ agents and the other is preferred by $m-1$ agents. The solution $a^*$ that minimizes the  social cost places the facility at point $(0,0)$ (preferred by $m$ agents)
resulting in a total cost of $\opt(\t) = m-1$. Now, take two different recommendations $\pred_1$ and $\pred_2$ at points $(-1,0)$ and $(1,0)$ respectively. The prediction error is the same for both points and it is equal to $\frac 1{m-1}$. However, any recommendation between $a^*$ and $\pred_2$ is almost optimal for large $m$, in contrast to $\pred_1$.  The quality of the recommendation captures this difference: the social cost for the two recommendations are $\C(\pred_1) = 3m-2$ and $\C(\pred_2) = m$, and therefore the 
quality of the recommendation for $\pred_1$ and $\pred_2$ are respectively  
$\err_1 = \frac{3m-2}{m-1}$ and $\err_2 = \frac{m}{m-1}$, which converge to $3$ and $1$ respectively as
$m$ grows. 

\begin{figure} [h]
\centering
\begin{tikzpicture}
  \def\d{2cm}
  \coordinate (start) at (-\d,0);
  \coordinate (middle) at (0,0);
  \coordinate (end) at (\d,0);

% point with \optt label
\filldraw[green] (start) circle (2pt) node[above, text=black] {$\pred_1$};

\filldraw[red] (middle) circle (2pt) node[above, text=black] {$\optt$};

% point with \pred label
\filldraw[green] (end) circle (2pt) node[above, text=black] {$\pred_2$};
  
\draw (middle) node[below] {$m$} node[] {\bm{$\times$}};
\draw (end) node[below] {$m-1$} node[] {\bm{$\times$}};
%\draw (\d/2, -0.1) node[below] {$1$};
\draw (start) -- (end);
\end{tikzpicture}
\caption{Quality of recommendation versus prediction error}
\label{fig:near_optimal}
\end{figure}

In Section~\ref{sec: eqalitarian}, we study the egalitarian cost and show that the Minimum Bounding Box Mechanism ~\ref{alg:MinimumBoundingBox}, defined by ~\citet{agrawal2022learning}, achieves an approximation ratio of $\err$, which combined with the robustness bound of \cite{agrawal2022learning} gives an overall approximation guarantee of $\min\{\err, \sqrt{2}+1\}$. In Section~\ref{sec: utilitarian} we focus on the utilitarian cost and show that the Coordinatewise Median Mechanism with predictions ~\ref{alg:CMP}, defined in ~\cite{agrawal2022learning}, achieves an approximation ratio of at most $\sqrt{2}\err$ which combined with the robustness bound of ~\cite{agrawal2022learning} gives an overall approximation guarantee of $\min\{\sqrt{2}\err,\frac{\sqrt{2\lambda^2+2}}{1-\lambda}\}$, where $\lambda\in [0,1)$ is a parameter that models the confidence of the designer on the recommendation; larger values of $\lambda$, correspond to increased confidence about the advice. Finally, in Section~\ref{sec: error} we compare
the bounds obtained as a function of the quality of recommendation to previously known results obtained as a function of the prediction error.

\subsection{Egalitarian Cost}\label{sec: eqalitarian}
The main result of this section is an approximation ratio of $\err$ for the egalitarian cost, by analyzing the Minimum Bounding Box mechanism defined in \cite{agrawal2022learning}. The robustness result for this mechanism \cite{agrawal2022learning}, gives a total approximation ratio of $\min\{\err, \sqrt{2}+1\}$, which is tight (Lemma \ref{lem:fl_mbb_lb}).

Intuitively, the Minimum Bounding Box mechanism works as follows\footnote{For completeness we define the Minimum Bounding Box mechanism in the Appendix (Section~\ref{sec:MBB}). We further refer the reader to \cite{agrawal2022learning} for the exact definition and also for the proof of strategyproofness and robustness.}: If the 
minimum rectangle that contains all the input points
$z_i, i\in\{ 1,\ldots, n\}$, contains the recommendation point $\pred$, then we output $\pred$. Otherwise, we select the boundary point with the minimum distance from $\pred$.

\begin{theorem}
\label{thm:MBBub}
    The Minimum Bounding Box mechanism is $\min\{\err, \sqrt{2}+1\}$-approximate.
\end{theorem}

\begin{proof}
\begin{align*}
\mech(\t,\pred) = \max_id(z_i, \f(\t, \pred)) 
          \leq \C(\t,\pred)
           = \err\opt(\t)
\end{align*}

The inequality follows from the fact that, whenever the prediction is
outside the minimum bounding box, the mechanism projects the
prediction on its boundaries, in a way that improves the egalitarian
loss compared to the initial prediction. When the prediction is inside
the bounding box, then $\f(\t, \pred) = \pred$ and the
inequality holds with equality. The term $(\sqrt{2}+1)$ follows from
the robustness guarantee proved in \cite{agrawal2022learning}. By
selecting the minimum of the two bounds, we get the approximation
above.
\end{proof}

Next, we show that the bound of Theorem~\ref{thm:MBBub} is tight.

\begin{lemma}\label{lem:fl_mbb_lb}
    For any $\err \leq \sqrt{2}+1$, there exists an instance where the approximation ratio of the Minimum Bounding Box mechanism is $\err$.
\end{lemma}

\begin{proof}
Consider the instance (used in \cite{agrawal2022learning}) of 3 agents with preferred locations $z_1=(0,1)$,  $z_2=(1,0)$ and $z_3=(-\frac 1{\sqrt{2}},-\frac 1{\sqrt{2}})$; see also Figure \ref{fig:worst_max}.
The optimal location is on $(0,0)$ with $\opt(\t)=1$. Consider the recommendation $\pred$ as a point on the line between $(0,0)$ and $(1,1)$, such that $d(\optt,\pred)=\err-1$; note that for any $\err \in [1,\sqrt{2}+1]$, there exists such a point. 
Since $\pred$ is inside the bounding box, the outcome of the mechanism remains on the same location, and therefore, $$\frac{\mech(\t,\pred)}{\opt(\t)}=C(\t,\pred)=d(z_3,\pred)=1+d(\optt,\pred)=\err\,.$$
\end{proof}

\begin{figure} [h]
\centering
\begin{tikzpicture}
\draw[->] (0,-1) -- (0,2) node[left] {$y$};
\draw[->] (-2,0) -- (2,0) node[below] {$x$};

% x points with no label
% Draw points
\draw (1,1) node[] {} node[above] {$(1,1)$};
\draw (0,1) node[] {$\bm{\times}$} node[left] {$(0,1)$};
\draw (1,0) node[] {$\bm{\times}$} node[below] {$(1,0)$};
\draw (-1/1.414,-1/1.414) node[] {$\bm{\times}$} node[left] {$(-\frac{1}{\sqrt{2}}, -\frac{1}{\sqrt{2}})$};

\filldraw[] (1,1) circle (1pt) node[] {};
\filldraw[red] (0,0) circle (2pt) node[above left, text=black] {$\optt$};

\filldraw[green] (0.7,0.7) circle (2pt) node[right, text=black] {$\pred = \f(\t,\pred)$};

\draw[dashed] (0,0) -- (1,1);
\end{tikzpicture}
\caption{Tight lower bound for the Minimum Bounding Box mechanism.}
\label{fig:worst_max}
\end{figure}
  
\begin{remark}
\label{rem:optimalityForMBB}
We remark that when $\f(\t, \pred)=\pred$, the upper bound of $\err$ is tight. In practice, this happens whenever the recommendation is inside the minimum bounding box defined by the agents' locations.
\end{remark}

\subsection{Utilitarian Cost}\label{sec: utilitarian}

Next, we show a $\sqrt{2}\err$ upper bound for the utilitarian cost by using the Coordinatewise Median with predictions mechanism defined in \cite{agrawal2022learning}.
This mechanism specifies a parameter $\lambda\in [0,1)$ which models
how much the recommendation is trusted. Intuitively,\footnote{We refer
  the reader to \cite{agrawal2022learning} for the exact definition
  and also for the proof of strategyproofness and robustness. For
  convenience, we also define it in the Appendix
  (Section~\ref{sec:MBB}).} the mechanism works as follows; it creates
$\lfloor\lambda n\rfloor$ copies of the recommendation
$\pred=(x_{\pred}, y_{\pred})$. Then, by treating each coordinate
separately, it selects the median point among $n + \lfloor\lambda n\rfloor$ in total
points; the $n$ actual bids $z_i=(x_i,y_i)$ and the $\lfloor\lambda n\rfloor$
copies of the recommendation. After calculating the medians $x_a$ and $y_a$ for each coordinate, it defines the outcome
to be $f(\t,\pred)=(x_a,y_a)$.

Before proving the main result, we show an
intermediate lemma. Given the location of a given facility  $a=(x_a,y_a)$, let $\C_x(\t, a)$ be the sum of distances between the $x$-coordinate of each point and $x_a$, i.e., $\C_x(\t, a)= \sum_id(x_i, a_x)$. Similarly, we define  $\C_y(\t, a)= \sum_id(y_i, a_y)$ for the $y$-coordinate.

\begin{lemma}\label{lem:1dcmm}
    Let $\f$ be the selection algorithm of the Coordinatewise Median with Predictions mechanism. Then $\C_x(\t, \f(\t, \pred)) \leq \C_x(\t, \pred)$. Similarly $\C_y(\t, \f(\t, \pred)) \leq \C_y(\t, \pred)$.
\end{lemma}

\begin{proof}
  In \cite{agrawal2022learning} they showed that for any two locations $a,a'$, it holds $d(f(\t, a),f(\t, a'))\leq d(a,a')$. Considering the single dimensional case, and the locations of $\optt$ and $\pred$, this implies that, $f(\t,\pred)$ lies between $a^*$ and $\pred$, since in the single dimensional case, $f(\t,\optt)=\optt$ (see Figure
  \ref{fig:procaccia}).  Then, by following a similar argument to the one used in \cite{procaccia2013approximate}, one can observe that $\C_x(\t, a)$ gets only higher values when the facility moves away from the median $a^*$; because more agents are harmed by the change and less of them benefit from it. 
\end{proof}

\begin{figure} [h]
\centering
\begin{tikzpicture}

\draw[->] (0,0) -- (10,0) node[right] {$x$};
\node[] at (1,0) {$\bm{\times}$};
\node[] at (2.5,0) {$\bm{\times}$};
\node[] at (4,0) {$\bm{\times}$};
\node[] at (6.5,0) {$\bm{\times}$};
\node[] at (6.9,0) {$\bm{\times}$};
\node[] at (9,0) {$\bm{\times}$};

\filldraw[red] (6,0) circle (2pt) node[below, text=black] {$\optt$};
\node[] at (6,0) {$\bm{\times}$};

\filldraw[blue] (5,0) circle (2pt) node[below, text=black] {$\f(\t,\pred)$};

\filldraw[green] (2,0) circle (2pt) node[below, text=black] {$\pred$};
\end{tikzpicture}
\caption{Projected locations on x-axis}
\label{fig:procaccia}
\end{figure}

\begin{lemma}\label{lem:distsum}
    For every point $z \in \mathbb{R}^2$ it holds that $d(z, f(\t, \pred)) \leq d(z,\pred) + d(z, f(\t))$, where $f(\t)$ is the coordinatewise median of points $\t$.
\end{lemma}
\begin{proof}
    By definition of the coordinatewise median with predictions mechanism, the output facility $a = f(\t, \pred)$ will be somewhere inside the rectangle defined by the opposite vertices $\pred$ and $f(\t)$. We take cases depending on where point $z$ lies (see Figure \ref{fig:improved_bound}):

    \begin{enumerate}
        \item If $z \in$ Region 1, then $d(z, a) \leq d_x(z, a) + d_y(z, a) \leq d_x(z, \pred) + d_y(z, m) \leq d(z,\pred) + d(z, f(\t))$, where $d_x, d_y$ are the $x$ and $y$ components of the according distances.
        \item If $z \in$ Region 2, then $d(z, a) \leq d(z, \pred)$
        \item If $z \in$ Region 3, then $d(z, a) \leq d(z, f(\t))$
    \end{enumerate}
\end{proof}

\begin{figure} [h]
\centering
\begin{tikzpicture}
\fill[gray!20] (1,1.2) -- (5,1.2) -- (5,3) -- (1,3);
\fill[gray!20] (1,1.2) -- (-1,1.2) -- (-1,-1) -- (1,-1);
\fill[red!20] (1,1.2) -- (1,3) -- (-1,3) -- (-1,1.2);
\fill[green!20] (1,1.2) -- (5,1.2) -- (5,-1) -- (1,-1);

% Draw rectangle
\draw[thick] (0,0) rectangle (4,2);
% Draw circles and labels
\filldraw[green] (0,2) circle (2pt) node[above, text=black] {$\pred$};
\filldraw[red] (4,0) circle (2pt) node[below, text=black] {$f(\mathbf{t})$};
\node[] at (1,1.2) {$\bm{\times}$};
\node[blue] at (1,1.5) {$\f(\t,\pred)$};
\draw[dashed] (1,-1) -- (1,3); % Vertical line passing through \times
\draw[dashed] (-1,1.2) -- (5,1.2); % Horizontal line passing through \times

\filldraw[black] (2,2.5) circle (2pt) node[above, text=black] {$z$};

\draw[dotted] (2,2.5) -- (1,1.2);
\draw[dotted] (2,2.5) -- (0,2);
\draw[dotted] (2,2.5) -- (4,0);

\node[draw,circle,inner sep=2pt] at (4.5,2.5) {1};
\node[draw,circle,inner sep=2pt] at (-0.7,-0.5) {1};
\node[draw,circle,inner sep=2pt] at (4.5,-0.5) {2};
\node[draw,circle,inner sep=2pt] at (-0.7,2.5) {3};

\end{tikzpicture}
\caption{$d(z, f(\t, \pred)) \leq d(z,\pred) + d(z, f(\t))$ for every agent location $z$}
\label{fig:improved_bound}
\end{figure}

By using the above together with simple geometric arguments, we are ready to present our main result on the approximation ratio of the Coordinatewise Median with Predictions mechanism.

\begin{theorem}
\label{thm:CMPub}
    The Coordinatewise Median with Predictions mechanism is $\min\{\sqrt{2}\err, \err + \sqrt{2}, \frac{\sqrt{2\lambda^2+2}}{1-\lambda}\}$-approximate.
\end{theorem}

\begin{proof}
  The robustness bound $\frac{\sqrt{2\lambda^2+2}}{1-\lambda}$ was shown in \cite{agrawal2022learning}. Next, we show that the approximation ratio is at most $\sqrt{2}\err$.
  
  \begin{align*}
    \mech(\t,\pred)  & = \sum_i{d(z_i, \f(\t, \pred))} \notag\\
              & \leq  \C_x(t, \f(\t, \pred)) + \C_y(t,\f(\t, \pred)) \\
        &\leq \C_x(\t, \pred) + \C_y(\t,\pred)\\ 
               &\leq \sqrt{2}\C(\t,\pred) \\
               &= \sqrt{2}\err\opt(\t) %\label{eq:CMPub}
    \end{align*}

    The mechanism approximation ratio is also at most $\err+\sqrt{2}$. To see this, we sum over all points $z_i$ and use Lemma \ref{lem:distsum}.

    \[\sum_i d(z, f(\t, \pred)) \leq \sum_i d(z,\pred) + \sum_i d(z, f(\t))\]
    \[\mech(\t,\pred) \leq \C(\t, \pred) + \C(\t, f(\t)) \leq \err \opt + \sqrt{2}\opt = (\err + \sqrt{2})\opt\]

    The second inequality follows from the definition of the quality of recommendation and the guarantee of the coordinatewise median without predictions.
\end{proof}

Next, we show that our analysis is tight, i.e., that there exists an instance where all of the inequalities above become equalities.

\begin{lemma}\label{lem:sum_exists}
    There exists an instance where the $\sqrt{2}\err$ upper bound on the approximation ratio of the Coordinatewise Median with Predictions mechanism is asymptotically tight.
\end{lemma}

\begin{proof}
Consider now the instance in Figure \ref{fig:worst_sum} (inspired by the worst case instance for the utilitarian objective~\cite{goel2023optimality}), where there is an agent at position $(-1,0)$, $m$ agents at position $(0,1)$ and $m-1$ agents at position $(1,0)$, for some $m\geq 2$. 
The optimal location is on $(0,1)$, with $\opt(\t) = \sqrt{2}m$. 
Let the recommendation $\pred$ be at location $(1,0)$, this gives $$\err = \frac {\C(\t,\pred)} {\opt(\t)} = \frac{\sqrt{2}m+2}{\sqrt{2}m} = 1+\frac{\sqrt{2}}m\,.$$
W.l.o.g. suppose that ties in the median algorithm on the $y$-coordinate are resolved by selecting the agent with the lowest value (otherwise, it suffices to create a symmetric instance by replacing point $(0,1)$ by $(0,-1)$).
Then, no matter what the value of $\lambda$ is, $\f(\t,\pred)$ goes to the origin according to Coordinatewise Median with Predictions mechanism, and so $\mech(\t,\pred)=2m$. Then, %approximation ratio of the mechanism is then 
$$\frac{\mech(\t,\pred)}{\opt(\t)}=\frac{2m}{\sqrt{2}m}=\sqrt{2}=\sqrt{2}\err - \frac{2}{m}\,,$$ 
that goes to $\sqrt{2}\err$ as $m$ grows. 
\end{proof}

\begin{figure} [h]
\centering
\begin{tikzpicture}
% Coordinate system
\draw[->] (0,0) -- (0,2) node[left] {$y$};
\draw[->] (-2,0) -- (2,0) node[below] {$x$};

% point with \optt label
\filldraw[red] (0,1) circle (2pt) node[above left, text=black] {$\optt(0,1)$};

% point with \pred label
\filldraw[green] (1,0) circle (2pt) node[above, text=black] {$\pred(1,0)$};

\filldraw[blue] (0,0) circle (2pt) node[below, text=black] {$\f(\t,\pred)$};

% x points with no label
\filldraw (0,1) node[] {$\bm{\times}$};
\filldraw (1,0) node[] {$\bm{\times}$};
\filldraw (-1,0) node[] {$\bm{\times}$} node[above] {$(-1,0)$};

\filldraw (0,1) node[right,font=\small] {$m$};
\filldraw (1,0) node[below,font=\small] {$m-1$};
\filldraw (-1,0) node[below,font=\small] {$1$};
\end{tikzpicture}
\caption{The instance showing the tight analysis for the Coordinatewise Median with Predictions mechanism.}
\label{fig:worst_sum}
\end{figure}

\subsection{Quality of recommendation vs prediction error}\label{sec: error}

At the beginning of this section, we provided an example indicating that the quality of recommendation captures the recommendation gap more accurately compared to the prediction error. In this subsection, we delve into further details on how the dependency on $\err$ can indeed result in a more precise analysis. We compare our bounds (from the previous subsections) concerning $\err$ with known bounds from the literature regarding the prediction error, for the same mechanisms.
The two mechanisms discussed in this section, Minimum Bounding Box and Coordinatewise Median with predictions, were defined in \cite{agrawal2022learning}. In that work, they consider a prediction error defined as $\eta=\frac{d(\optt, \pred)}{\opt(\t)}$ for the egalitarian social cost and $\eta=\frac{nd(\optt, \pred)}{\opt(\t)}$ for the utilitarian social cost. 

In \cite{agrawal2022learning}, they show an upper bound of 
$\min\{\eta + 1,\sqrt{2}+1\}$ for the approximation ratio of the Minimum Bounding Box mechanism, and an upper bound of 
$\min\{\frac{\sqrt{2\lambda^2+2}}{1+\lambda}+\eta,\frac{\sqrt{2\lambda^2+2}}{1-\lambda}\}$
for the approximation ratio of the Coordinatewise Median with Predictions mechanism.

We first establish that $\err \leq \eta+1$ holds for both objectives. This implies that our result showing the Minimum Bounding Box mechanism to be $\err$-approximate (refer to Theorem~\ref{thm:MBBub}), also implies the $\eta+1$ upper bound as shown in \cite{agrawal2022learning}. However, we demonstrate that there exist instances where our $\err$ bound strictly outperforms the $\eta+1$ bound, further justifying our choice to evaluate the recommendation using $\err$.\footnote{We note that both bounds, $\err$ and $\eta+1$, are tight due to the instance shown in Figure~\ref{fig:worst_max}.} As for the Coordinatewise Median with Predictions mechanism, we show that our $\sqrt{2}\err$ upper bound (refer to Theorem~\ref{thm:CMPub}) contributes to the overall understanding of the Coordinatewise Median with Predictions mechanism performance by providing an instance where $\sqrt{2}\err$ is strictly more accurate than $\frac{\sqrt{2\lambda^2+2}}{1+\lambda}+\eta$.

Next we show a relation between $\err$ and $\eta$
for both egalitarian and utilitarian social cost functions.

\begin{lemma}\label{lem:etas}
    $\err \leq \eta+1$ for both the egalitarian and the utilitarian objective.
\end{lemma}

\begin{proof}
    Regarding the egalitarian cost, let $z_{\max}^{\pred}$ be the point with the maximum distance from $\pred$ and let $z_{\max}^{\optt}$ be the one with the maximum distance from the optimal allocation $\optt$.
    By using the triangle inequality we get $d(\optt, \pred) \geq d(\pred, z_{\max}^{\pred}) - d(\optt, z_{\max}^{\pred}) \geq d(\pred, z_{\max}^{\pred}) - d(\optt, z_{\max}^{\optt})=\C(\t,\pred)-\opt(\t)$. The second inequality holds due to the fact that $z_{\max}^{\optt}$ has the maximum distance from $\optt$. After normalizing by $\opt(\t)$, we derive the desired inequality, $\eta \geq \err - 1$.

    Regarding the utilitarian cost, by applying the triangle inequality to all $z_i$ locations, it holds that 
    $$\err = \frac{\C(\t,\pred)}{\opt(\t)} = \frac{\sum_i{d(z_i,\pred)}}{\opt(\t)} \leq \frac{\sum_i{d(\optt,\pred)}+\sum_id(\optt,z_i)}{\opt(\t)} = \frac{nd(\optt,\pred)+\opt(\t)}{\opt(\t)} =\eta+1\,.$$
\end{proof}

The above lemma shows that, regarding the egalitarian social cost, our bound of $\err$ is (weakly) more accurate than the known bound of $\eta+1$ by \cite{agrawal2022learning}. However, by Remark~\ref{rem:optimalityForMBB}, the $\err$ bound strictly outperforms the $(1+\eta)$ bound in any instance where the recommendation is inside the minimum bounding box and the $(1+\eta)$ bound is not tight. 
Next we provide such an instance, meaning that our analysis is indeed more accurate.

\begin{lemma}
\label{lem:betterEgalitarian}
    For the egalitarian social cost, there exists an instance where $\err < \eta+1$.
\end{lemma}
\begin{proof}
Consider the instance in Figure \ref{fig:sum_1} of 4 agents with preferred locations $z_1=(1,1)$, $z_2=(-1,1)$, $z_3=(-1,-1)$ and $z_4=(1,-1)$. The optimal location is on position $(0,0)$, with $\opt(\t) = \sqrt{2}$. Suppose that the recommendation $\pred$ is on location $(0,1)$. The outcome of the Minimum Bounding Box mechanism sets the facility on the recommended location, as it belongs inside the minimum bounding box surrounding the agents. Then, $\C(\t,\pred)=\sqrt{5}$ and $\err = \sqrt{2.5} \approx 1.581$, while $\eta= 1/\sqrt{2} \approx 0.707 > \err - 1$. The approximation ratio is $\err$ and our bound is thus tight, whereas the $\eta+1$ is not.
\end{proof}

\begin{figure}[h]
\centering
\begin{tikzpicture}
\draw[->] (0,-1.3) -- (0,2) node[left] {$y$};
\draw[->] (-2,0) -- (2,0) node[below] {$x$};

\draw (1,1) node[] {$\bm{\times}$} node[above right] {$(1,1)$};
\draw (-1,1) node[] {$\bm{\times}$} node[above left] {$(-1,1)$};
\draw (-1,-1) node[] {$\bm{\times}$} node[left] {$(-1,-1)$};
\draw (1,-1) node[] {$\bm{\times}$} node[right] {$(1,-1)$};

\filldraw[red] (0,0) circle (2pt) node[above, text=black] {$\optt$};

\filldraw[green] (0,1) circle (2pt) node[above, text=black] {$\pred = \f(\t,\pred)$};

\draw[dashed] (1,-1) -- (0,1) node[midway, above, black] {$\sqrt{5}$};
\draw[dashed] (1,-1) -- (0,0) node[midway, left, black] {$\sqrt{2}$};

\end{tikzpicture}
\caption{An instance satisfying the statements of Lemmas~\ref{lem:betterEgalitarian} and \ref{lem:betterUtalitarian}}
\label{fig:sum_1}
\end{figure}

Regarding the utilitarian social cost, we present an instance where our bound of $\sqrt{2}\err$ outperforms the $\frac{\sqrt{2\lambda^2+2}}{1+\lambda} + \eta$ bound. %For the second instance, we show that our bound is tight.

\begin{lemma}
\label{lem:betterUtalitarian}
    For the utilitarian social cost, there exists an instance where $\sqrt{2}\err < \frac{\sqrt{2\lambda^2+2}}{1+\lambda}+\eta$
\end{lemma}
\begin{proof}
Consider again the instance in Figure \ref{fig:sum_1}. The optimal location is on position $(0,0)$, with $\opt(\t) = 4\sqrt{2}$.
W.l.o.g. suppose that ties in the median algorithm on the $y$-coordinate are resolved by selecting the agent with the highest value.\footnote{If ties were resolved differently, then it suffices to create a symmetric instance by replacing the position $(0,1)$ of the recommendation by $(0,-1)$.} 
No matter what the value of parameter $\lambda$ is, $\f(\t, \pred)=\pred$. From this fact, it follows that $\C(\t,\pred)=2\sqrt{5}+2$ and $\err = \frac{\sqrt{5}+1}{2\sqrt{2}} \approx 1.144$, while $\eta= 1 / \sqrt{2} = 0.707$. Our upper bound is $\sqrt{2}\err = 1.618$, while the upper bound on \cite{agrawal2022learning} is $\frac{\sqrt{2\lambda^2+2}}{1+\lambda} + \eta > 1 + \eta = 1.707 > 1.618$. The $\sqrt{2}\err$ bound is thus closer to the approximation ratio of the mechanism for this instance (which is $\frac{\mech(\t,\pred)}{\opt(\t)}=\frac{\sqrt{5}+1}{2\sqrt{2}} = 1.144$).
\end{proof}

Further comparison between the two error functions is shown in the experimental Section \ref{sec:experiments}.

\subsection{Experiments}\label{sec:experiments}

In this section, the quality of recommendation $\err$ is compared to the error $\eta$ defined in \cite{agrawal2022learning} as a function of the performance ratio of the Coordinatewise Median with Predictions mechanism. We run the specific mechanism on several real datasets with various predictions $\pred$ as recommendations, and compare the behavior of the two error functions.\footnote{https://github.com/yannisvl/OutRecom} Experiments are executed on an Intel Core i7-6500U CPU 2.50GHz-2.59 GHz, with 8 GB of RAM.

Predictions are created uniformly forming a grid on the bounding box defined by the input points. The optimal solution is found using Weiszfeld's algorithm. The confidence parameter of the Coordinatewise Median Mechanism was set to high $\lambda$ values, for which the mechanism highly depends on the prediction, in order to better compare the two errors as the prediction varies. Datasets used are location based timestamped data from real scenarios with potential facility location applications. Similarly to previous work, we probe the Twitter dataset \cite{DBLP:conf/www/ChanGS18} containing social network's post locations, used for the online facility location with predictions problem in \cite{DBLP:conf/nips/Cohen-AddadHPSS19}, the Brightkite and Gowalla datasets from the SNAP dataset collection \cite{jure2014snap}, also examined in \cite{DBLP:conf/nips/AlmanzaCLPR21}, and the Earthquake \cite{earthquakeRef} and Autotel \cite{autotelRef} datasets. The Earthquake dataset contains coordinates of past earthquakes and the Autotel dataset contains shared cars locations from the namesake company.

\begin{figure} [h]
    \centering
    \begin{subfigure}{0.45\textwidth}
        \includegraphics[width=\linewidth]{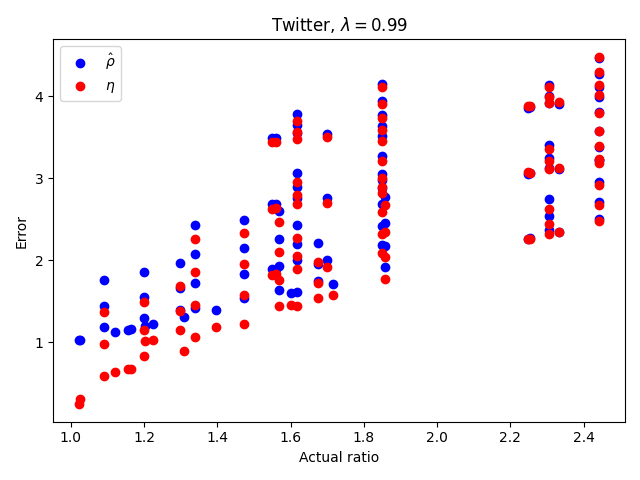}
    \end{subfigure}
    \begin{subfigure}{0.45\textwidth}
        \includegraphics[width=\linewidth]{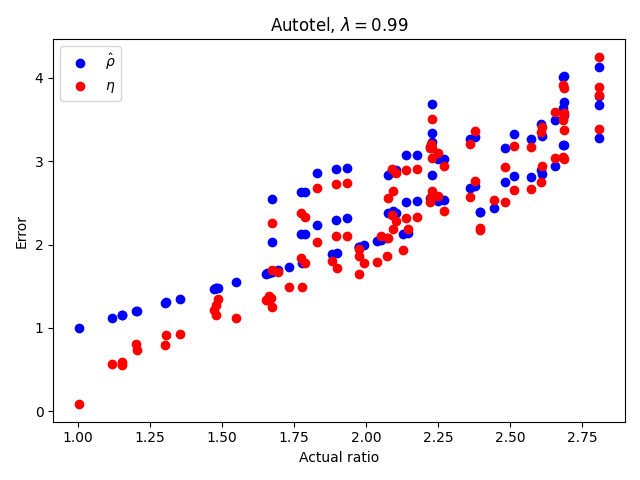}
    \end{subfigure}
    \newline
    \begin{subfigure}{0.45\textwidth}
        \includegraphics[width=\linewidth]{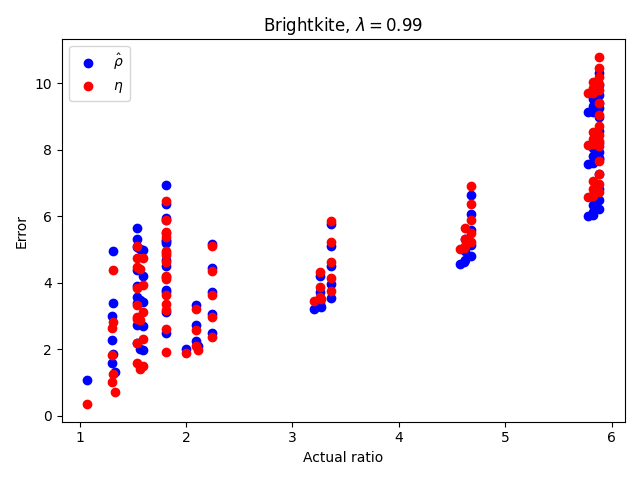}
    \end{subfigure}
    \caption{Comparison between $\err$, $\eta$ as a function of the Coordinatewise Median with Predictions mechanism set to $\lambda = 0.99$ for Twitter, Brightkite and Autotel datasets.}
    \label{fig:first_experiments}
\end{figure}

From the experiments it can be seen that both of the errors exhibit an increasing behavior as a function of the Coordinatewise Median mechanism's ratio. However, $\eta$ exceeds the quality of recommendation $\err$ as the ratio grows (see Figure \ref{fig:first_experiments}). This can be seen more clearly in the Gowalla dataset (see Figure \ref{fig:both_confidence}), and especially in the example of the Earthquake dataset (see Figure \ref{fig:both_confidence}), where we see that the $\eta$ error can grow very large comparing to the quality of recommendation $\err$, showing results similar to Example \ref{fig:near_optimal} in real-world datasets.

\begin{figure} [h]
    \centering    
    \begin{subfigure}{0.45\textwidth}
        \includegraphics[width=\linewidth]{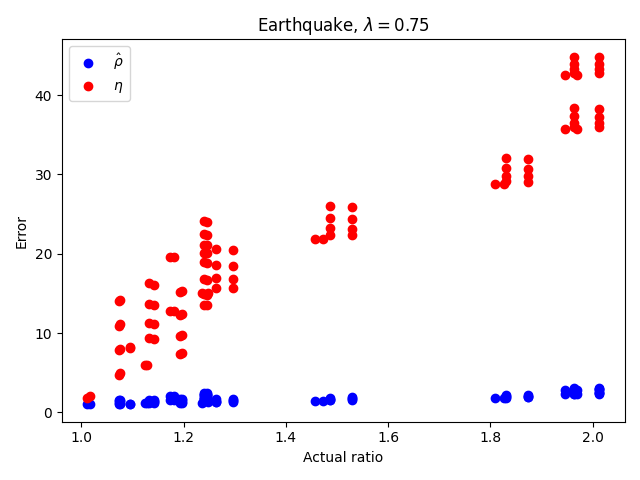}
    \end{subfigure}
    \begin{subfigure}{0.45\textwidth}
        \includegraphics[width=\linewidth]{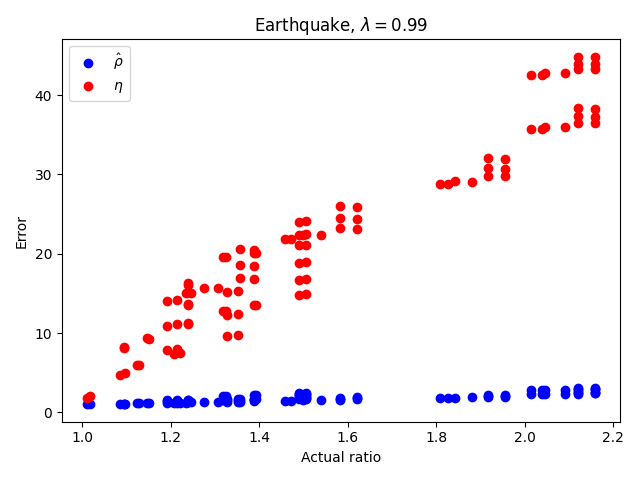}
    \end{subfigure}
    \newline
    \begin{subfigure}{0.45\textwidth}
        \includegraphics[width=\linewidth]{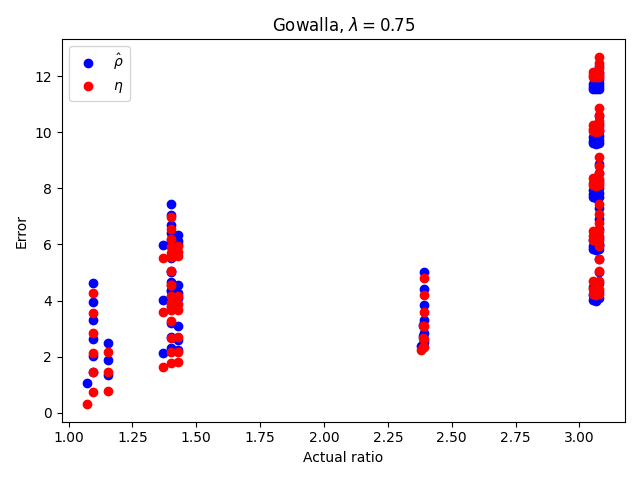}
    \end{subfigure}
    \begin{subfigure}{0.45\textwidth}
        \includegraphics[width=\linewidth]{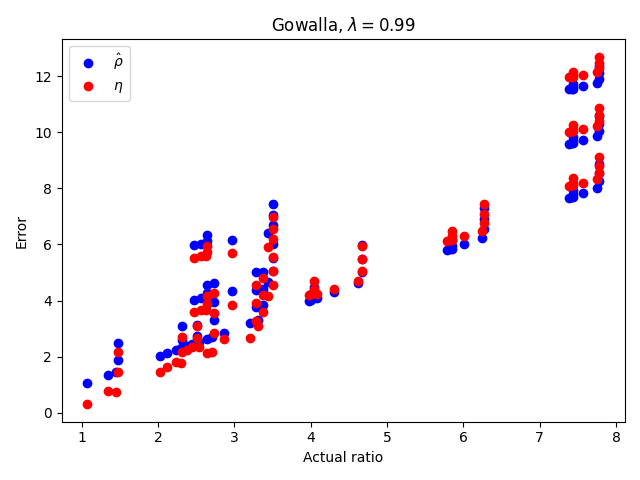}
    \end{subfigure}
    \caption{Comparison between $\err$, $\eta$ as a function of the Coordinatewise Median with Predictions mechanism set to $\lambda = 0.75$ and $\lambda = 0.99$ for Gowalla and Earthquake datasets.}
    \label{fig:both_confidence}
\end{figure}

%% file: scheduling.tex
In this section, we study strategyproof mechanisms for the \textit{makespan minimization scheduling problem}. In this problem, we have a set $N$ of $n$ unrelated machines (considered as the agents) and a set $M$ of $m$ jobs. Each machine $i$ has a (private) cost $t_{ij}$ for each job $j$, which corresponds to the processing time of job $j$ in machine $i$. Since we consider only strategyproof mechanisms, each machine $i$ declares their {\em true} cost $t_{ij}$ for each job $j$; let $t_i=(t_{i1}, \ldots, t_{im})$. The goal of the mechanism is to process the declarations $\t=(t_1,\ldots,t_n)$ of the machines and subsequently determine both an allocation $a(\t)$ of the jobs to the machines and a payment scheme $p(\t)=(p_1(\t)\ldots,p_n(\t))$, where $p_i(\t)$ is given to each machine $i$ for processing their allocated jobs.
An allocation is given by a vector $a = (a_1, \ldots, a_n)$, where $a_i = (a_{i1}, \ldots, a_{im})$, and $a_{ij}$ is set to 1 if job $j$ is assigned to machine $i$ and 0 otherwise. An allocation $a$ is feasible if each job is allocated to exactly one machine, i.e., $\sum_{i\in N}a_{ij}=1$, for all $j\in M$, and $\sum_{i\in N, j\in M}a_{ij}=m$; we denote by $\mathcal{A}$ the set of all feasible allocations.

The cost experienced by each machine $i$ under an allocation $a$ is the total cost of all jobs assigned to it: $t_i(a) = t_i(a_i)=\sum_{j \in M}{t_{ij}a_{ij}} = t_i \cdot a_i$. The private objective of each machine $i$ is to maximize their utility $u_i(\t)=p_i(\t)-t_i(a(\t))$. In the strategyproof mechanisms that we consider here, this happens when each machine declares its true cost. 
The social cost function that is usually used in this problem in order to evaluate the quality of an allocation $a$, is the maximum cost among all machines, which is known as the makespan: $\C(\t, a) = \max_i t_i(a)$.

We assume that the mechanism is provided with a recommendation $\pred\in \mathcal{A}$, which can be seen as a suggestion on how to allocate the jobs to the machines. For a given $\t$ we denote by $\optt(\t)$ the optimal allocation minimizing the social cost function, i.e., $\optt(\t)\in \arg\min_{a\in \mathcal{A}}\C(\t, a)$, and by $\opt(\t)$ the minimum social cost, i.e., $\opt(\t)=\C(\t, \optt(\t))$. We measure the quality of the recommended outcome with  $\err(\t)$, which is defined as the approximation ratio $C(\t,\pred)/\opt(\t)$ and measures the approximation that we would achieve if we selected the recommended allocation $\pred$.
In the notation of $\optt$ and $\err$, we drop the dependency on $\t$ when it is clear from the context.

In the remainder of this section, we introduce a strategyproof mechanism that we call AllocationScaledGreedy (Mechanism~\ref{alg:AllocationScaledGreedy}). We prove that, given a confidence parameter $1\leq \beta \leq n$, it exhibits $(\beta+1)$-consistency and $\frac{n^2}{\beta}$-robustness (Theorem~\ref{the:all_greedy}). Next, we investigate the smoothness of this mechanism and demonstrate that its approximation ratio is upper bounded by $\min\{(\beta + 1)\err, n+\err, \frac{n^2}{\beta}\}$, which is asymptotically tight (Theorem~\ref{thm:ASGsmoothness}). Furthermore, we establish that, when provided with the outcome as advice, it is impossible to achieve a better consistency-robustness trade-off than the AllocationScaledGreedy mechanism within the class of weighted VCG mechanisms (Theorem~\ref{the:asg_opt}).

\subsection{AllocationScaledGreedy Mechanism}

In this subsection, we introduce a strategyproof mechanism called AllocationScaledGreedy. We demonstrate that it achieves a $(\beta+1)$-consistency (more precisely, $(\frac{n-1}n \beta +1)$-consistency which converges to $\beta+1$ for large $n$) and a $\frac{n^2}{\beta}$-robustness, where $\beta$ is a confidence parameter ranging from $1$ to $n$, with $1$ corresponding to full trust and $n$ corresponding to mistrust. 
For $\beta=n$, which can be interpreted as ignoring the recommendation, the AllocationScaledGreedy mechanism corresponds to the VCG mechanism; in that case, consistency and robustness bounds coincide, giving an $n$-approximation (same as VCG). For simplicity, we consider $(\beta+1)$-consistency which is approximately the consistency guarantee for large values of $n$. Regarding the smoothness of our mechanism, we prove an asymptotically tight approximation ratio of $\min\{(\beta + 1)\err, n+\err, \frac{n^2}{\beta}\}$.

\paragraph{AllocationScaledGreedy} The mechanism sets a weight $r_{ij}$ for every machine $i$ and every job $j$ based on the recommendation $\pred$. $r_{ij}$ is set to 1 wherever $\pred_{ij}=1$, and $\frac{n}{\beta}$ wherever $\pred_{ij}=0$, for some $\beta \in [1,n]$. It then decides the allocation by running the weighted VCG mechanism for each job $j$ separately, and by using $r_{ij}$ as the (multiplicative) weight of machine $i$, i.e., each job $j$ is allocated to some machine in $\arg\min_{i}\{r_{ij}t_{ij}\}$ that we denote by $i_j$.

\begin{algorithm} [h]
\caption{The AllocationScaledGreedy mechanism}
\begin{algorithmic}[1]
\Require instance $\t \in \mathbb{R}^{n \times m}$, recommendation $\pred \in \mathbb{R}^{n \times m}$
\Ensure $a$
\State $r_{ij} \gets 1$ if $\pred_{ij}=1$, $\frac{n}{\beta}$ otherwise, $(\beta \in [1,n])$
\State $i_j \gets \argmin_{i}\{r_{ij}t_{ij}\}$%, ties in favor of $\hat{\imath}_j$
\State if $i=i_j$ then $a_{ij}=1$ else $a_{ij}=0$, for each $(i, j) \in N \times M$
\end{algorithmic}
\label{alg:AllocationScaledGreedy}
\end{algorithm}

\begin{remark}
    We remark that the AllocationScaledGreedy mechanism for $\beta=1$ is a simplification of the SimpleScaledGreedy mechanism of \cite{balkanski2022strategyproof}. 
    In \cite{balkanski2022strategyproof}, it is assumed that the mechanism is equipped with predictions of the entire cost matrix $\hat{t}_{ij}$, for every machine-job pair. The SimpleScaledGreedy mechanism utilizes this information to define weights $r_{ij}$ that may take values in the range $[1,n]$. In contrast, AllocationScaledGreedy uses weights with values only $1$ or $n$, for $\beta=1$. Notably, despite the limited information available to AllocationScaledGreedy, both mechanisms share the same consistency and robustness, but SimpleScaledGreedy lacks the nice property of being smooth, as for a very small prediction error, the approximation ratio has a large discontinuity gap (see Lemma~\ref{lem:SSGgap}) as opposed to AllocationScaledGreedy (see Theorem~\ref{thm:ASGsmoothness}). SimpleScaledGreedy served as an intermediate step in \cite{balkanski2022strategyproof} in the design of the more sophisticated mechanism ScaledGreedy, (which again relies heavily on the prediction of the entire cost matrix) which achieves the best of both worlds, constant consistency and $O(n)$-robustness. It's worth noting that for similar reasons, ScaledGreedy is not smooth either. 
\end{remark}

\begin{lemma}
\label{lem:SSGgap}
    There exists an instance, for which the approximation ratio of both the SimpleScaledGreedy and ScaledGreedy mechanisms \cite{balkanski2022strategyproof} is $\Omega(n)$ for prediction error $\eta$ arbitrarily close to 1.
\end{lemma}
\begin{proof}
    Consider the case of Figure \ref{fig:jump_example}, with $n$ machines and $n-1$ jobs. All costs in the predicted instance are $1$ except for the costs of machine $n$ which are $1+\epsilon$ for each job. The real instance differs only on the costs of machine $n$ which are $1-\epsilon$, meaning arbitrarily close to the predicted ones as $\epsilon$ goes to $0$. The SimpleScaledGreedy and the ScaledGreedy mechanisms~\cite{balkanski2022strategyproof} set all the weights to $r_{ij}=1$. This results in an $\Omega(n)$ approximation for a very small error $\eta = \max\{\frac{r_{ij}}{\hat{r}_{ij}},\frac{\hat{r}_{ij}}{r_{ij}}\} = \frac{1+\epsilon}{1-\epsilon} \approx 1$. This approximation becomes linear for an error arbitrarily close to $1$, in contrast to the constant bound guaranteed by the consistency analysis of \cite{balkanski2022strategyproof}. 
\end{proof}

\begin{figure} [h]
    \centering
    \begin{subfigure}{0.45\textwidth}
        \centering
        \begin{tabular}{c|cccc}
            $\hat{t}_{ij}$&1&2&$\dots$&$n-1$\\
            \hline
             1&1 & 1 & $\dots$ &  1\\
            2&1 & 1 & $\dots$ &  1\\
            $\vdots$&$\vdots$ & $\vdots$ & $\ddots$ & $\vdots$\\
            $n-1$&1 & 1 & $\dots$ & 1 \\
            $n$&$1+\epsilon$ & $1+\epsilon$ & $\dots$ & $1+\epsilon$ \\
        \end{tabular}
    \end{subfigure}
    \begin{subfigure}{0.45\textwidth}
        \centering
        \begin{tabular}{c|cccc}
            $t_{ij}$&1&2&$\dots$&$n-1$\\
            \hline
            1&1 & 1 & $\dots$ &  1\\
            2&1 & 1 & $\dots$ &  1\\
            $\vdots$&$\vdots$ & $\vdots$ & $\ddots$ & $\vdots$\\
            $n-1$&1 & 1 & $\dots$ & 1 \\
            $n$&$1-\epsilon$ & $1-\epsilon$ & $\dots$ & $1-\epsilon$ \\
        \end{tabular}
    \end{subfigure}
    \caption{Smoothness violation for the SimpleScaledGreedy and ScaledGreedy mechanisms. The predicted costs are given on the left. Given the prediction, both mechanisms set all weights to $1$. If the actual costs are as on the right, the approximation ratio for both mechanisms is approximately $n-1$ for a prediction error arbitrarily close to 1. Moreover, the quality of recommendation, if considering the optimal outcome for the predicted input as the recommendation, is 1.}
    \label{fig:jump_example}
\end{figure}

\begin{theorem}\label{the:all_greedy}
    The AllocationScaledGreedy mechanism is $\left(\frac{n-1}{n}\beta+1\right)$-consistent and $\frac{n^2}{\beta}$-robust.
\end{theorem}

\begin{proof}
    First, we establish that the AllocationScaledGreedy mechanism achieves $\left(\frac{n-1}{n}\beta+1\right)$-consistency. Let $\t$ be any cost vector for the machines and consider a perfect recommendation, i.e., $\pred=\optt$. For {\em any} machine $i$, we divide the set of jobs into two groups, $M_i \cap \hat{M}_i$ and $M_i \setminus \hat{M}_i$, where $M_i$ is the set of jobs assigned to machine $i$ by the mechanism, i.e., $M_i=\{j: a_{ij}=1\}$, and $\hat{M}_i$ is the recommended assignment,  i.e., $\hat{M}_i=\{j: \pred_{ij}=1\}$. Then, it holds that,

     \begin{align}
     \sum_{j \in M_i \cap \hat{M}_i}{t_{ij}} \leq \sum_{j \in \hat{M}_i}{t_{ij}} = t_i(\pred) \leq \C(\t,\pred) = \opt(\t) \,. \label{eq:consistency1}
     \end{align}
     
    Next, let $\hat{\imath}_j$ be the machine that receives job $j$ in the recommended allocation, i.e., $\pred_{\hat{\imath}_j j}=1$. For any job $j \in M_i \setminus \hat{M}_i$, according to the AllocationScaledGreedy mechanism it holds that $r_{ij}t_{ij}\leq r_{\hat{\imath}_j j}t_{\hat{\imath}_j j}$. It further holds that $i\neq \hat{\imath}_j$, which in turn gives $r_{ij}=\frac{n}{\beta}$ and $r_{\hat{\imath}_jj}=1$. Overall, we get: 

    \[
       \sum_{j \in M_i \setminus \hat{M}_i}{t_{ij}} \leq \sum_{j \in M_i \setminus \hat{M}_i}{\frac{r_{\hat{\imath}_j j}}{r_{ij}}t_{\hat{\imath}_j j}} = \frac{\beta}{n}\sum_{j \in M_i \setminus \hat{M}_i}{t_{\hat{\imath}_j j}} \leq 
       \frac{\beta}{n}\sum_{i' \in N\setminus\{i\}}\sum_{j \in \hat{M}_{i'}}{t_{\hat{\imath}_j j}}\,.
    \]

    It also holds that the average completion time of machines in $N\setminus\{i\}$ is bounded by the maximum completion time in the same set, which in turn is bounded by the makespan of the mechanism:
    
    \[
       \frac{\beta}{n}\sum_{i' \in N\setminus\{i\}}\sum_{j \in \hat{M}_{i'}}{t_{\hat{\imath}_j j}}\leq \frac{n-1}{n}\beta \cdot \max_{i' \in N\setminus\{i\}}\sum_{j \in \hat{M}_{i'}}{t_{\hat{\imath}_j j}}
       \leq \frac{n-1}{n}\beta \C(\t,\pred) = \frac{n-1}{n}\beta \opt(\t) \,.
    \]

    Overall, by combining these two facts:
    
    \begin{align} \label{eq:consistency2}
       \sum_{j \in M_i \setminus \hat{M}_i}{t_{ij}} &\leq \frac{\beta}{n}\sum_{i' \in N\setminus\{i\}}\sum_{j \in \hat{M}_{i'}}{t_{\hat{\imath}_j j}} \leq \frac{n-1}{n}\beta \opt(\t) \,. 
    \end{align}

   Recall that the above inequalities hold for any machine $i$, so by summing them up we get 
    $$\mech(\t,\pred) = \mech(\t,\optt) = \max_{i\in N} \sum_{j \in M_i}{t_{ij}} \leq \left(\frac{n-1}{n}\beta + 1\right)\opt(\t),$$
    which gives a $(\frac{n-1}{n}\beta + 1)$-consistency. 
    
    To establish robustness, let $\t$ be any cost vector for the machines and $\pred$ be any recommendation. Let $i_j^*$ be the machine that receives job $j$ in the optimal allocation, i.e., $\optt_{i_j^* j}=1$. According to the AllocationScaledGreedy mechanism, for every job $j \in M_i$, it holds that $t_{ij} \leq r_{ij}t_{ij}\leq r_{i_j^* j}t_{i_j^* j}\leq \frac n{\beta} t_{i_j^* j}$. Based on that, we get the desired result regarding robustness as follows:

    \[\mech(\t,\pred) = \max_{i \in N}{\sum_{j \in M_i}{t_{ij}}} \leq \sum_{i \in N}\sum_{j \in M_i}{t_{ij}} \leq \frac{n}{\beta}\sum_{i \in N}\sum_{j \in M_i}{t_{i_j^* j}} = \frac{n}{\beta}\sum_{i \in N}{t_{i}\optt_i(\t)} \leq \frac{n^2}{\beta}\max_{i \in N}\{t_{i}\optt_i(\t)\} = \frac{n^2}{\beta}\opt(\t)\,.\]
\end{proof}

In the following theorem, we show the smoothness result for the AllocationScaledGreedy mechanism; we show a tight approximation ratio depending on $\err$.

\begin{theorem}
\label{thm:ASGsmoothness}
    The AllocationScaledGreedy mechanism is at most  $\min\{(\beta + 1)\err, n+\err, \frac{n^2}{\beta}\}$-approximate and  this bound is asymptotically tight.
\end{theorem}

We prove this theorem in the following two lemmas, where we show that $\min\{(\beta + 1)\err, n+\err, \frac{n^2}{\beta}\}$ is an upper bound (Lemma~\ref{lem:asg_ub}) and $\min\{\frac {n-1}n\beta \err, \frac{n+\err-1}2, \frac{n^2-1}{2\beta}\}$ is a lower bound (Lemma~\ref{lem:asg_lb}) on the approximation ratio of the AllocationScaledGreedy mechanism.

\begin{lemma}\label{lem:asg_ub}
    The AllocationScaledGreedy is at most $\min\{(\beta + 1)\err, n+\err, \frac{n^2}{\beta}\}$-approximate.
\end{lemma}

\begin{proof}    
The upper bound of $\frac{n^2}{\beta}$ is trivially derived by Theorem~\ref{the:all_greedy} where we showed that the AllocationScaledGreedy mechanism is $\frac{n^2}{\beta}$-robust.  

We first show the $(\beta + 1)\err$ upper bound. We use inequalities \eqref{eq:consistency1} and \eqref{eq:consistency2}, before applying the final assumption that $\pred = \optt$, which give an upper bound on the makespan of the mechanism's outcome that depends on the makespan of the recommended allocation:

\[\mech(\t,\pred) \leq (\beta + 1)\C(\t,\pred)=(\beta + 1)\err \opt(\t)\,.\]

We next give an alternative analysis that results in the $n+\err$ upper bound. For any $\t$ and any machine $i$, let $M_i$ be the set of jobs assigned to machine $i$ by the mechanism, i.e., $M_i=\{j: a_{ij}=1\}$, and $\hat{M}_i$ is the recommended assignment,  i.e., $\hat{M}_i=\{j: \pred_{ij}=1\}$. Then, by the definition of $\err$:

\[ \sum_{j \in M_i \cap \hat{M}_i}{t_{ij}} \leq \sum_{j \in \hat{M}_i}{t_{ij}} \leq \C(\t,\pred) = \err \opt(\t) \,.\]
     
Next, let $i_j^*$ be the machine that receives job $j$ in the optimal allocation, i.e., $\optt_{i_j^* j}=1$. For any job $j \in M_i \setminus \hat{M}_i$, the AllocationScaledGreedy mechanism ensures that $r_{ij}t_{ij}\leq r_{i_j^* j}t_{i_j^* j}$ and $r_{ij}=\frac{n}{\beta}$. Consequently, $t_{ij}\leq t_{i_j^* j}$, irrespective of whether $r_{i_j^*j}$ equals to $\frac{n}{\beta}$ or $1$. Therefore, it holds that:

\[
   \sum_{j \in M_i \setminus \hat{M}_i}{t_{ij}} \leq \sum_{j \in M_i \setminus \hat{M}_i}{t_{i_j^*j}} \leq \sum_{j \in M}{t_{i_j^*j}} = \sum_{i \in N} t_i(\optt) \leq n \cdot \max_{i \in N} t_i(\optt) = n \cdot \opt(\t)\,.
\]

By summing up the 2 inequalities, we get the $(n + \err)$ upper bound.
\end{proof}

In the following lemma, we prove that the upper bound of Lemma~\ref{lem:asg_ub} is asymptotically tight.

\begin{lemma}\label{lem:asg_lb}
The AllocationScaledGreedy is at least $\min\{\frac {n-1}n\beta \err, \frac{n+\err-1}2, \frac{n^2-1}{2\beta}\}$-approximate.
\end{lemma}

\begin{proof}
Consider first the case where $\err \leq  \frac{n}{\beta}$. We provide an instance where the AllocationScaledGreedy mechanism has a lower bound of $\frac{n-1}n\beta\err$. In this instance there are $n$ machines and $n-1$ jobs, and the recommended allocation $\pred$ is that each machine $i<n$ is allocated job $i$, i.e., $\pred_{ij}=1$, only if $i=j$, see also Figure (\ref{fig:lb_case1}). Following the AllocationScaledGreedy mechanism, $r_{ij}=1$ for $i=j$ and $r_{ij}=\frac{n}{\beta}$ otherwise. Suppose now that the true types/costs of the machines are as follows (see also Figure \ref{fig:lb_case1}): 
\begin{itemize}
    \item $t_{ij}=\err$, for any $i < n$ and $j=i$ (for the job allocated to $i$ according to the recommendation),
    \item $t_{nj}=\frac{\beta\err}n-\epsilon$, for any $j$ and for some arbitrarily small $\epsilon>0$,
    \item $t_{ij}=1$, for any other case.
\end{itemize}
It is easy to see that the optimal allocation is to allocate each job to a different machine\footnote{The only case this does not hold is when $\beta \err \frac{n-1}{n} < 1$. In that case, the optimal solution assigns all jobs to machine $n$, AllocationScaledGreedy is at least $1$-approximate and the lemma holds trivially.}, resulting in $\opt(\t)=1$, and moreover, $\C(\t,\pred)=\err$, so the quality of recommendation is $\err$. The allocation induced by the AllocationScaledGreedy mechanism is to allocate all jobs to machine $n$, since for any job $j$, $r_{nj}t_{nj}<\err$, $r_{jj}t_{jj}=\err$, and $r_{ij}t_{ij}=\frac{n}{\beta}\geq \err$, for any $i\notin \{j, n+1\}$. 
This means that $\mech(\t,\pred)=\frac{n-1}n \beta\err-(n-1)\epsilon$, which goes to $\frac{n-1}n\beta\err$ as $\epsilon$ goes to $0$. 

\begin{figure} [h]
    \centering
    \begin{subfigure}{0.45\textwidth}
        \centering
        \begin{tabular}{c|cccc}
            $\pred_{ij}$&$1$&$2$&$\dots$&$n-1$\\
            \hline
            $1$&1 & 0 & $\ldots$ & 0 \\
            $2$&0 & 1 & $\ldots$ & 0\\
            $\vdots$&$\vdots$& $\vdots$ & $\ddots$  & $\vdots$\\
            $n-1$&0 & 0 & $\ldots$ & 1\\
            $n$&0 & 0 & $\ldots$ & 0\\
        \end{tabular}
    \end{subfigure}
    \begin{subfigure}{0.45\textwidth}
        \centering
        \begin{tabular}{c|cccc}
            $t_{ij}$&1&2&$\dots$&$n-1$\\
            \hline
            1&$\hat{\rho}$ & 1 & $\ldots$ & 1 \\
            2&1 & $\hat{\rho}$ & $\ldots$ & 1\\
            $\vdots$&$\vdots$& $\vdots$ & $\ddots$  & $\vdots$\\
            $n-1$&1 & 1 & $\ldots$ & $\hat{\rho}$\\
            $n$&$\frac{\beta \hat{\rho}}{n}-\epsilon$ & $\frac{\beta \hat{\rho}}{n}-\epsilon$ & $\ldots$ & $\frac{\beta \hat{\rho}}{n}-\epsilon$\\
        \end{tabular}
    \end{subfigure}
    \caption{The recommendation $\pred$ (on the left) and the true cost vector (on the right) used in the first case ($\err \leq \frac{n}{\beta}$) in the proof of Lemma~\ref{lem:asg_lb}. The AllocationScaledGreedy mechanism assigns all jobs to agent $n$.}
    \label{fig:lb_case1}
\end{figure}

Now consider the case where $\frac{n}{\beta}<\err\leq \frac{n^2-n}{\beta}$. We provide an instance where the AllocationScaledGreedy mechanism has a lower bound of $\frac{n+\err-1}2$. In this instance there are $n$ machines and $k + n-1$ jobs, where $k = \left\lceil \frac{\beta \err}{n}\right\rceil$; it is true that $1 < \frac{\beta \err}{n}\leq n-1$. 
The recommended allocation $\pred$ is that each machine $i<n$ is allocated job $i$ and machine $n$ is allocated jobs $\{n,\ldots, n+k-1\}$, see also Figure \ref{fig:lb_case21}.

Following the AllocationScaledGreedy mechanism, $r_{ij}=1$ for $i=j<n$, and for $i=n$ and $j\geq n$;  $r_{ij}=\frac{n}{\beta}$ otherwise. Suppose now that the true types/costs of the machines are as follows (see also Figure \ref{fig:lb_case22}), where $\epsilon>0$ is arbitrarily small: 
\begin{itemize}
    \item $t_{ij}=2\hat{\rho}$, for any $i =j< n$,
    \item $t_{nj}=1-\epsilon$, for $j<n$,
    \item $t_{nj}=\frac{n}\beta-\epsilon$, for $j\geq n$ ,
    \item $t_{ij}=1$, for any other case.
\end{itemize}

It is easy to see that the optimal solution is equal to 2, by assigning at most 2 jobs to each of the first $n-1$ machines\footnote{Actually, it is at most 2, because for the case $k=2$, it is possible to assign 2 of the first $n-1$ jobs to machine $n$ resulting in $\opt = 2-2\epsilon$}. The quality of the recommendation is $\frac{\C(\t,\pred)}{\opt(\t)}=\err$, (for machine $n$, $t_n(\pred)=\left\lceil \frac{\beta \err}{n}\right\rceil \left(\frac n{\beta}-\epsilon\right)\leq \left(\frac{\beta \err}{n}+1\right)\frac n{\beta}=\err+\frac n{\beta}<2\err$). Now, the mechanism assigns all jobs to machine $n$ resulting in: 

$$\mech(\t,\pred)=\left\lceil \frac{\beta \err}{n}\right\rceil \left(\frac n{\beta}-\epsilon\right) +n-1 \geq \frac{\beta \err}{n} \left(\frac n{\beta}-\epsilon\right) +n-1 = \err  +n-1 -\frac{\beta \err}{n}\epsilon\geq\left(\frac{\err  +n-1}2 -\frac{\beta \err}{2n}\epsilon\right)\opt(\t)\,,$$ 

that converges to $\frac{\err +n-1}2$ as $\epsilon$ goes to $0$.

If $\err>\frac{n^2-n}{\beta}$, there exists a lower bound of $\frac{n^2-1}{2\beta}$ that is derived similarly to the second case (Figure~\ref{fig:lb_case22}) by setting $k=n-1$.
\end{proof}

\begin{figure}[h]
    \centering
    \begin{tabular}{c|*{4}{c}:*{4}{c}}
        $\pred_{ij}$&1&2&$\dots$&$n-1$&$n$&$n+1$&$\dots$&$n+k-1$\\
        \hline
        1&1&0&$\ldots$&0&  0&0&$\ldots$&0\\
        2&0&1&$\ldots$&0&  0&0&$\ldots$&0 \\
        $\vdots$&$\vdots$&$\vdots$&$\ddots$&$\vdots$&$\vdots$&$\vdots$&$\ddots$ &$\vdots$ \\
        $n-1$&0&0&$\ldots$&1&0&0&$\ldots$&0 \\
        \hdashline
        $n$&0&0&$\ldots$&0&1&1 & $\ldots$ & 1 \\
    \end{tabular}
    \caption{The recommendation $\pred$ used for the case that $\err>\frac{n}{\beta}$ in the proof of Lemma~\ref{lem:asg_lb}. The parameter $k$ is set to $k = \min\left\{n-1,\left\lceil \frac{\beta \err}{n}\right\rceil\right\}$.}
    \label{fig:lb_case21}
\end{figure}

\begin{figure}[h]
    \centering
    \begin{tabular}{c|*{4}{c}:*{4}{c}}
        $t_{ij}$&1&2&$\dots$&$n-1$&$n$&$n+1$&$\dots$&$n+k-1$\\
        \hline
        1&$2\err$&1&$\ldots$&1&  1&1&$\ldots$&1\\
        2&1&$2\err$&$\ldots$&1&  1&1&$\ldots$&1 \\
        $\vdots$&$\vdots$&$\vdots$&$\ddots$&$\vdots$&$\vdots$&$\vdots$&$\ddots$ &$\vdots$ \\
        $n-1$&1&1&$\ldots$&$2\err$&1&1&$\ldots$&1 \\
        \hdashline
        $n$&$1-\epsilon$&$1-\epsilon$&$\ldots$&$1-\epsilon$&$\frac{n}{\beta}-\epsilon$&$\frac{n}{\beta}-\epsilon$ & $\ldots$ & $\frac{n}{\beta}-\epsilon$ \\
    \end{tabular}
    \caption{The true cost vector $\t$ used for the case that $\err>\frac{n}{\beta}$ in the proof of Lemma~\ref{lem:asg_lb}. The parameter $k$ is set to $k = \min\left\{n-1,\left\lceil \frac{\beta \err}{n}\right\rceil\right\}$.}
    \label{fig:lb_case22}
\end{figure}

\subsection{Mechanism Optimality}
In this subsection, we provide general impossibility results for the class of weighted VCG mechanisms\footnote{Technically, weighted VCG mechanisms choose weights $r_i$ for each machine $i$, rather than the more general case of choosing $r_{ij}$ for each machine $i$ and job $j$, that we consider here. In scheduling, where the valuation domain is additive, jobs can be grouped into clusters, and a distinct VCG mechanism can be applied to each cluster. The composition of these mechanisms remains strategyproof for additive domains. The extreme (and more general) case considered here is to cluster the jobs into $m$ clusters.}, the most general known class of strategyproof mechanisms for multi-dimensional mechanism design settings, such as the scheduling problem. We prove that it is impossible to improve upon the AllocationScaledGreedy mechanism, given the recommended outcome. More specifically, there is no weighted VCG mechanism with $\beta$-consistency that can achieve a robustness better than $\Theta(\frac{n^2}{\beta})$, highlighting the optimality of AllocationScaledGreedy in this class of mechanisms.

\begin{theorem}\label{the:asg_opt}
    Given any recommendation $\pred$, any weighted VCG mechanism that is $\beta$-consistent, must also be $\Omega(\frac{n^2}{\beta})$-robust, for any $2 \leq \beta \leq n$.
\end{theorem}

{\em Proof sketch.} We provide a proof sketch of Theorem~\ref{the:asg_opt} before stating the complete proof. We will consider instances with $n$ machines and $n^2$ jobs. Let a $\beta$-consistent weighted VCG mechanism and a recommendation $\pred$ that assigns every $n$ jobs to a distinct machine. Focusing on each machine $i$, we specify the cost vector $\t$, such that the optimal allocation matches $\pred$. The costs are such that the mechanism must assign each job $j$ either to machine $i$ or to machine $\hat{\imath}_j$ that receives job $j$ in $\pred$. Machine $i$ should not receive many jobs, otherwise $\beta$-consistency is violated. Consequently, there are many (approximately $\frac{n^2}2$) weights $r_{ij}$ with value much higher comparing to the weight $r_{\hat{\imath}_jj}$, i.e., $\frac{r_{ij}}{r_{\hat{\imath}_jj}}\geq \frac{n}{2\beta}$. 

Since this is true for each machine $i$, there exists a machine $\hat{\imath}$, such that, focusing only on the $n$ jobs  that $\hat{\imath}$ receives in $\pred$, there exist approximately $\frac{n^2}2$ jobs with value much higher (comparing to $\hat{\imath}$) among all machines. Then it holds that we can assign approximately $\frac{n}2$ jobs to distinct machines such that those machines have high-valued weight for their assigned job; let $J$ be the set of those jobs. We finally consider the instance where each of those machines has a cost of $1$ for their assigned job and sufficiently high cost\footnote{We choose $\infty$ cost for clarity, in fact it suffices to choose instead $t_{ij}>\frac{\min_{i'}\{r_{i'j}t_{i'j}\}}{r_{ij}}$, such that the mechanism does not allocate job $j$ to machine $i$.} for any other job in $J$, machine $\hat{\imath}$ has a cost slightly less than $\frac{n}{2\beta}$ for jobs in $J$, and all other machines have infinite cost for jobs in $J$. The cost for any other job that does not belong to $J$ is $0$ for any machine. In this instance $\t$, $\opt(\t)=1$, but the mechanism allocates all jobs of $J$ to machine $\hat{\imath}$, resulting in $\mech(\t,\pred)$ being approximately $\frac{n^2}{4\beta}$. Hence, any $\beta$-consistent weighted VCG mechanism is $\Omega(\frac{n^2}{\beta})$-robust.

\begin{proof}
Consider any strategyproof mechanism that is $\beta$-consistent and selects the allocation $a(\t)$ for each cost vector $\t$. First, we will establish some properties regarding the weights to ensure that $\beta$-consistency is satisfied.
Now, let's examine an instance with $n$ machines and $n^2$ jobs, along with a recommendation $\pred$ that allocates every group of $n$ jobs to a different machine, e.g., $\pred_{ij}=1$, if $j \in \{(i-1)n+1, \ldots, (i-1)n+n\}$ and  $\pred_{ij}=0$ otherwise, as illustrated in Figure \ref{tab:n2robust_xij}.

\begin{figure} [h]
    \centering
    \begin{tabular}{c|ccc:ccc:c:ccc:ccc}
        $\pred_{ij}$&1&$\dots$&$n$&$n+1$&&$2n$&$\dots$&$n^2-2n+1$&$\ldots$&$n^2-n$&$n^2-n+1$&$\dots$&$n^2$\\
        \hline
        1&1&$\dots$&1&0&$\dots$&0&$\dots$&0&$\dots$&0&0&$\dots$&0\\
        2&0&$\dots$&0&1&$\dots$&1&$\dots$&0&$\dots$&0&0&$\dots$&0\\
        $\vdots$&$\vdots$&$\ddots$&$\vdots$&$\vdots$&$\ddots$&$\vdots$&$\vdots$&$\vdots$&$\ddots$&$\vdots$&$\vdots$&$\ddots$&$\vdots$\\
        $n-1$&0&$\dots$&0&0&$\dots$&0&$\dots$&1&$\dots$&1&0&$\dots$&0\\
        $n$&0&$\dots$&0&0&$\dots$&0&$\dots$&0&$\dots$&0&1&$\dots$&1\\
    \end{tabular}
    \caption{Recommendation $\pred$ used in the proof of Theorem \ref{the:asg_opt}}
    \label{tab:n2robust_xij}
\end{figure}

A weighted VCG mechanism chooses the weights $r_{ij}$ for each pair $i,j$. The fact that it is $\beta$-consistent, means that for any true cost vector $\t$ for which $\optt(\t) =\pred$, $\mech(\t,\pred)\leq \beta\opt(\t)$. Suppose now the instance that satisfies $\optt(\t) =\pred$, where $t_{ij}=n$ for $\pred_{ij}=1$, $t_{1j}=2\beta$ for any $j>n$, and $t_{ij}=\infty$\footnote{We choose $\infty$ cost for clarity, in fact it suffices to choose $t_{ij}>\frac{\min_{i'}\{r_{i'j}t_{i'j}\}}{r_{ij}}$, such that the mechanism does not allocate job $j$ to machine $i$.} for any other case, see also Figure~\ref{tab:n2robust_pred}. In $\t$, the optimal  allocation matches the recommended allocation $\pred$ and $\opt(\t)=n^2$. For each job $j$, let $\hat{\imath}_j$ be the machine that receives job $j$ in the recommended allocation, i.e., $\pred_{\hat{\imath}_jj}=1$. For each job $j$, we will normalize the weights $r_{ij}$ against the weight of machine $\hat{\imath}_j$: let  $w_{ij}=\frac{r_{ij}}{r_{\hat{\imath}_jj}}$.  

The mechanism will allocate each job $j\leq n$, to machine 1, and each job $j>n$ either to machine $\hat{\imath}_j$, or to machine 1. More specifically, for any $j>n$, if $r_{1j}t_{1j}<r_{\hat{\imath}_jj}t_{\hat{\imath}_jj}$, or alternatively if $w_{ij}<\frac{n}{2\beta}$, then $j$ is given to machine 1.  
Since the mechanism is $\beta$-consistent, $t_i(a)$ must not exceed $\beta \opt(\t) = \beta n^2$. Let's say that $k$ tasks with value $2\beta$ are assigned to machine $i$, then $\beta$-consistency requires that $n^2 + 2\beta k \leq \beta n^2 $, which implies $k \leq \frac{\beta - 1}{2\beta}n^2.$
Therefore, at most $\frac{\beta - 1}{2\beta}n^2$ jobs (apart from the first $n$) may be given to machine 1, so for at least $$n(n-1)-\frac{\beta - 1}{2\beta}n^2 = \frac{\beta + 1}{2\beta}n^2-n$$
jobs, the weight for agent $i$ is at least $\frac n{2\beta}$.

It turns out that this lower bound on the number of jobs with weights of at least $\frac n{2\beta}$, holds for {\em any} single machine, since we can consider any instance as described above where the "role" of machine 1 is given to any other machine.

\begin{figure} [h]
    \centering
    \begin{tabular}{c|ccc:ccc:c:ccc:ccc}
        $t_{ij}$&1&$\dots$&$n$&$n+1$&&$2n$&$\dots$&$n^2-2n+1$&$\ldots$&$n^2-n$&$n^2-n+1$&$\dots$&$n^2$\\
        \hline
        1&$n$&$\dots$&$n$&$2\beta$&$\dots$&$2\beta$&$\dots$&$2\beta$&$\dots$&$2\beta$&$2\beta$&$\dots$&$2\beta$\\
        2&$\infty$&$\dots$&$\infty$&$n$&$\dots$&$n$&$\dots$&$\infty$&$\dots$&$\infty$&$\infty$&$\dots$&$\infty$\\
        $\vdots$&$\vdots$&$\ddots$&$\vdots$&$\vdots$&$\ddots$&$\vdots$&$\vdots$&$\vdots$&$\ddots$&$\vdots$&$\vdots$&$\ddots$&$\vdots$\\
        $n-1$&$\infty$&$\dots$&$\infty$&$\infty$&$\dots$&$\infty$&$\dots$&$n$&$\dots$&$n$&$\infty$&$\dots$&$\infty$\\
        $n$&$\infty$&$\dots$&$\infty$&$\infty$&$\dots$&$\infty$&$\dots$&$\infty$&$\dots$&$\infty$&$n$&$\dots$&$n$\\
    \end{tabular}
    \caption{The true cost vector $\t$ that is used in the proof of Theorem \ref{the:asg_opt} to show the existence of a large number of high-valued weights for machine $1$ compared to the weight of the machine receiving the job in $\pred$, for any weighted VCG mechanism that is $\beta$-robust.}
    \label{tab:n2robust_pred}
\end{figure}

In summary, among all the $w_{ij}$ for any $1\leq i\leq n$ and $1\leq j\leq n^2$, there are at least $\frac{\beta + 1}{2\beta}n^3-n^2$ of them with value at least $\frac n{2\beta}$. Suppose now, that we partition the $n^2$ jobs into $n$ blocks, where each block $B_i$ is formed as the jobs $B_i=\{(i-1)n+1, \ldots, (i-1)n+n\}$. By the pigeonhole principle, there must exist one block, let it be $B_{\ell}$, with at least $\frac{\beta + 1}{2\beta}n^2-n$ weights of value at least $\frac{n}{2\beta}$, i.e., there exist at least $\frac{\beta + 1}{2\beta}n^2-n$ pairs $(i,j)$, with $1\leq i\leq n$ and $j\in B_{\ell}$, such that $w_{ij}\geq \frac{n}{2\beta}$.

    We aim to identify the longest possible diagonal within this block with weights of at least $\frac{n}{2\beta}$, (by allowing the rearrangement of machines and jobs). This problem can be considered as a \textit{Maximum Matching} of a bipartite graph $G$ that is constructed as follows. $G$ comprises $n$ vertices on one side $A$, representing the $n$ machines, and $n$ vertices on the other side $B$, representing the jobs in $B_{\ell}$. An edge $(i,j)$, for $i\in A$ and $j\in B$ is included if and only if $w_{ij}\geq \frac{n}{2\beta}$.
    We seek a maximum matching on this bipartite graph with $\frac{\beta + 1}{2\beta}n^2-n$ edges.
    By an extension of Konig's theorem~\cite{kHonig1931grafok} and by the fact that the maximum possible degree is $n$, the size of the maximum matching ($MM$) is given by:
    
    \[|MM| \geq \frac{|E(G)|}{\Delta(G)} \geq \frac{\frac{\beta + 1}{2\beta}n^2-n}{n} = \frac{\beta + 1}{2\beta}n - 1\,,\] 
    where $E(G)$ is the set of the edges in $G$ and $\Delta(G)$ is the maximum degree in $G$. 

    Let $q=\frac{\beta + 1}{2\beta}n - 1$ and suppose that we reorder the machines and the jobs such that each machine $i\leq q$ is matched to the $i^{th}$ job under $MM$ and machine $q+1$ becomes the machine $\ell$, i.e., the machine that receives all those jobs in $\pred$. Then for any $i\leq q$, it holds that $\frac{r_{ii}}{r_{(q+1) i}}=w_{ii}\geq \frac{n}{2\beta}$. 
    Focusing on those first $q+1$ machines and the first $q$ jobs, we consider the following instance (see also Figure~\ref{fig:weights_on_actual}):
    \begin{itemize}
        \item $t_{ii}=1$, for any $i\leq q$,
        \item $t_{(q+1) j}=\frac n{2\beta}-\epsilon$, for any $j\leq q$ and an arbitrarily small $\epsilon>0$, 
        \item $t_{ij}=0$, for any $i$ and any $j>q$,
        \item $t_{ij}=\infty$ for any other case.
    \end{itemize}
   
    It is easy to see that $\opt(\t)=1$ by allocating job $i$ to machine $i$. However, $r_{ii}t_{ii}\geq r_{(q+1)i}\frac{n}{2\beta}>r_{(q+1)i}t_{(q+1)i}$, which means that the mechanism allocates all the first $q$ jobs to machine $q+1$, resulting in $\frac{qn}{2\beta}$ approximation ratio as $\epsilon$ goes to $0$. Since $q>\frac n2 -1$, the approximation ratio is $\Omega(\frac{n^2}{\beta})$.
    
\end{proof}

\begin{figure} [h]
    \centering
    \begin{tabular}{c|ccccccc}
        $t_{ij}$&1&2&$\dots$&$q$&$q+1$&$\dots$&$n$\\
        \hline
        1&1 & $\infty$ & $\dots$ & $\infty$&$0$&$\dots$&$0$ \\
        2&$\infty$ & 1 & $\dots$ & $\infty$&$0$&$\dots$&$0$ \\
        $\vdots$&$\vdots$ & $\vdots$ & $\ddots$ & $\vdots$& $\vdots$ & $\ddots$ & $\vdots$\\
        $q$&$\infty$ & $\infty$ & $\dots$ & 1&$0$&$\dots$&$0$ \\
        $q+1$&$\frac{n}{2\beta}-\epsilon$ & $\frac{n}{2\beta}-\epsilon$ & $\dots$ &$\frac{n}{2\beta}-\epsilon$&$0$&$\dots$&$0$ \\
        $q+2$&$\infty$ & $\infty$ & $\dots$ & $\infty$&$0$&$\dots$&$0$ \\
        $\vdots$&$\vdots$ & $\vdots$ & $\ddots$ & $\vdots$& $\vdots$ & $\ddots$ & $\vdots$\\
        $n$&$\infty$ & $\infty$ & $\dots$ & $\infty$&$0$&$\dots$&$0$ \\
    \end{tabular}
    \caption{In this instance it holds that $w_{ii}\geq \frac{n}{2\beta}$, for $i\leq q$, and $w_{(q+1)j}=1$, for $j\leq q$. The weighted VCG mechanism allocates the first $q$ jobs to machine $q+1$, resulting in an approximation ratio of $\Omega(\frac{n^2}{\beta})$, which gives the lower bound on the robustness for Theorem~\ref{the:asg_opt}.}
    \label{fig:weights_on_actual}
\end{figure}

%% file: houseAllocation.tex
In this section, we study the \textit{house allocation} problem, where
the goal is to assign $n$ houses to a set of $n$ agents. A feasible
allocation $a = (a_1, \ldots, a_n)\in \mathcal{A}$ is a matching where
$a_i$ denotes the house allocated to agent $i$. For outcome $a$, agent $i$
has a private value $t_i(a) = t_i(a_i)$ for house $a_i$, having as a
private objective to maximize this value. A matching is evaluated by
the social welfare $\W(\t, a)=\sum_{i}{t_i(a_i})$ and the designer's
goal is to select the matching that maximizes this quantity. We assume
that we are given as recommendation a matching $\pred$. We denote by
$\err(\t) = \frac{\opt(\t)}{\W(\t,\pred)}$ the approximation we would
obtain by selecting the recommended matching as output, and we use $\err$ for
simplicity when its clarity is evident.

We will define a strategyproof mechanism that is
$\min\{\err, n\}$-approximate for \textit{unit-range} valuations and
$\min\{\err, n^2\}$-approximate for \textit{unit-sum} valuations. The
idea is to use the recommended matching as an initial endowment, and
then run the Top Trading Cycles (TTC) mechanism~\cite{shapley1974cores}. 

\paragraph{TTC mechanism}The TTC mechanism unfolds in rounds within a directed graph, where vertices represent agents (along with their endowments) and directed edges depict the current top preferences of agents. In each $k^{\text{th}}$ round, every remaining agent points to the owner of their preferred house among the available options. Given the finite number of agents, at least one cycle emerges, with agents pointing to each other or, alternatively, an agent pointing to oneself. Each agent within the cycle is then assigned the house of the agent they point to, and they are subsequently removed from the market with this allocation. If at least one agent remains unassigned, the process proceeds to the next round.

\begin{algorithm} [h]
\caption{TTC with recommended endowment}
\begin{algorithmic}[1]
\Require types $\t \in \mathbb{R}^{n\times n}$, recommended matching $\pred$
\Ensure matching $a$
\State use $\pred$ as initial endowment
\State Run TTC
\end{algorithmic}
\label{alg:ttc}
\end{algorithm}

First, we prove a tight bound for the robustness of the mechanism by
proving tight bounds for the TTC mechanism. More specifically, we
observe that the approximation ratio of the TTC mechanism is
$\Theta(n)$ for unit-range valuations and $\Theta(n^2)$ for unit-sum
valuations, resulting directly in the following theorem regarding robustness. 

\begin{lemma}\label{lem:ttc_ur}
The TTC mechanism is $\Theta(n)$-approximate for unit-range valuations.
\end{lemma}

\begin{proof}
Consider the instance in Figure \ref{fig:lb_house_allocation} (on the left), where agent 1 has preferences $(1-\epsilon, 1, 0, 0, \dots)$ over items 1, 2, 3 etc. For the rest, agent $i$ has preferences 1 for item $i$ and $\epsilon$ for item $(i+1)\mod n$. Now consider the worst case initial endowment $a = (2,3,4, \dots, n, 1)$. In the first round, the first agent will point to himself as he is assigned his highest preference and agent $i$ will point to agent $(i-1)$, who owns the highest preference of agent $i$. Thus, no cycle of size greater than 1 is formed and no changes will be made. Hence, the TTC mechanism results in a $\mech(\t,\pred) = 1 + (n-1)\epsilon$ social welfare, but the optimal solution is $\opt(\t) = n - \epsilon$. For $\epsilon \to 0$, it follows that $\frac{\opt(\t)}{\mech(\t,\pred)}$ is $\Omega(n)$.

Consider now any instance with unit-range valuations. In the first round, there is at least one agent that gets his best preference which will be $1$ (due to unit-range valuations) and hence $\mech(\t,\pred) \geq 1$. In addition, $\opt(\t)\leq n$, as each agent gets at most a value of $1$. Therefore, $\frac{\opt(\t)}{\mech(\t,\pred)} \leq n$.

Combining the above, we get the desired result.
\end{proof}

\begin{figure} [h]
    \centering
    \begin{subfigure}{0.35\textwidth}
        \centering
        \begin{tabular}{c|cccccc}
            $t_{ij}$&1&2&3&4&$\dots$&$n$\\
            \hline
            1&$1-\epsilon$& 1&0&0& $\dots$& 0\\
            2&0 & 1 & $\epsilon$&0&$\dots$& 0 \\
            3&0 & 0 & 1 &$\epsilon$&$\dots$&0\\
            4&0 & 0 & 0 &1 &$\dots$&0\\
            $\vdots$&$\vdots$& $\vdots$& $\vdots$&$\vdots$ & $\ddots$ & $\vdots$\\
            $n$&$\epsilon$& 0 & 0&0 & $\dots$ & 1 \\
        \end{tabular}
    \end{subfigure}%
    \begin{subfigure}{0.65\textwidth}
        \centering
        \begin{tabular}{c|cccccc}
            $t_{ij}$&1&2&3&4&$\dots$&$n$\\
            \hline
            1&$\frac{1}{n}-\epsilon$& $\frac{1}{n}+\epsilon$&$\frac{1}{n}$&$\frac{1}{n}$& $\dots$& $\frac{1}{n}$\\
            2&0 & $1-\epsilon$ & $\epsilon$&0&$\dots$& 0 \\
            3&0 & 0 & $1-\epsilon$ & $\epsilon$&$\dots$&0\\
            4&0 & 0 & 0 &1 &$\dots$&0\\
            $\vdots$&$\vdots$& $\vdots$&$\vdots$&$\vdots$ & $\ddots$ & $\vdots$\\
            $n$&$\epsilon$& 0 & 0& 0 & $\dots$ & $1-\epsilon$ \\
        \end{tabular}
    \end{subfigure}%
    \caption{TTC Lower bound instances (unit-range valuations on the left and unit-sum valuations on the right). Rows represent agents and columns represent houses. After running TTC with initial endowment $a = (2,3,4, \dots, n, 1)$, the first agent gets item 2, whereas the rest get item $(i+1)$mod$n$, $i>1$.}
    \label{fig:lb_house_allocation}
\end{figure}

\begin{lemma}\label{lem:ttc_us}
The TTC mechanism is $\Theta(n^2)$-approximate for unit-sum valuations.
\end{lemma}

\begin{proof}
Consider the instance in Figure \ref{fig:lb_house_allocation} (on the right) where the first agent has the following preferences $(\frac{1}{n}-\epsilon, \frac{1}{n}+\epsilon, \frac{1}{n}, \frac{1}{n}, \dots, \frac{1}{n})$. For the rest, agent $i$ has preferences $1-\epsilon$ for item $i$ and $\epsilon$ for item $(i+1)\mod n$. Consider also the initial endowment $a = (2,3,4, \dots, n, 1)$. Again, the first agent will point to himself and agent $i$ will point to agent $(i-1)$. No exchanges will be made, resulting in a social welfare $\mech(\t,\pred) = \frac{1}{n} + n\epsilon$. The optimal solution is again the main diagonal with a social welfare of $\opt(\t)=\frac{1}{n}- n\epsilon + n-1$. For $\epsilon \to 0$, it follows that the approximation ratio $\frac{\opt(\t)}{\mech(\t,\pred)}$ is $\Omega(n^2)$.

Consider now any instance with unit-sum valuations. In the first round, there is at least one agent that gets his best preference which will be at least $\frac{1}{n}$ (due to unit-sum valuations) and hence $\mech(\t,\pred) \geq \frac{1}{n}$. In addition, $\opt(\t) \leq n$, as each agent gets at most a value of $1$. Therefore, $\frac{\opt(\t)}{\mech(\t,\pred)} \leq n^2$.

Combining the above, we get the desired result.
\end{proof}

In total, we get the following theorem:

\begin{theorem}\label{cor:ttc_rob}
The TTC with recommended endowment mechanism is $n$-robust for unit-range valuations and $n^2$-robust for unit-sum valuations.
\end{theorem}

We note that these robustness bounds are the best possible, as we
discussed in Section~\ref{sec:related}, and match the general lower
bounds of all deterministic mechanisms.

We are now ready to show a general approximation ratio for our
mechanism.

\begin{lemma}\label{lem:ttc_appr}
    The TTC with recommended endowment mechanism is strategyproof, $O(\min\{\err, n\})$-approximate for unit-range valuations and $O(\min\{\err, n^2\})$-approximate for unit-sum valuations.
\end{lemma}

\begin{proof}
  It is known that the TTC mechanism is strategyproof~\cite{roth1982incentive}
  for the \textit{housing market} problem, which is the house
    allocation problem with a matching given as an initial endowment.

  By choosing the recommended matching as initial endowment for TTC, we
  make sure that the approximation ratio of the final matching never
  exceeds the initial approximation guarantee $\err$. This is because
  during the execution of TTC,  the social welfare can only improve, after the
  completion of each round, since each agent either remains with their
  initial endowment or strictly improves their value. We thus guarantee
  a $\err$-approximation. The robustness bounds follow from Theorem~\ref{cor:ttc_rob}. 
\end{proof}

Similarly to lemmas \ref{lem:ttc_ur}, \ref{lem:ttc_us}, it is shown in lemmas \ref{lem:ttc_lb_ur}, \ref{lem:ttc_lb_us} that these approximation ratios are tight for the TTC mechanism.

\begin{lemma}\label{lem:ttc_lb_ur}
    The TTC with recommended endowment mechanism is $\Omega(\min\{\err, n\})$-approximate for unit-range valuations.
\end{lemma}
\begin{proof}
    Consider the instance in Figure \ref{fig:lb_ttc_with_endowment} (on the left), where agent 1 has preferences $(1-\epsilon, 1, 0, 0, \dots)$ over items 1, 2, 3 etc. For the rest, agent $i$ has preferences 1 for item $i$ and $x = \frac{n-\err}{\err(n-1)}$ for item $(i+1)\mod n$. Clearly $x \geq 0$ as $\err \geq 1$ and $\err \leq n$ (because of the TTC mechanism lower bound \ref{lem:ttc_lb_ur}), while $x\leq 1$ holds trivially. Now consider the recommended initial endowment $\pred = (2,3,4, \dots, n, 1)$. In the first round, the first agent will point to himself as he is assigned his highest preference and agent $i$ will point to agent $(i-1)$, who owns the highest preference of agent $i$. Thus, no cycle of size greater than 1 is formed and no changes will be made. Hence, the TTC mechanism results in a $\mech(\t,\pred) = W(\t, \pred) = 1 + (n-1)x$ social welfare, but the optimal solution is $\opt(\t) = n - \epsilon$ according to $\opt = (1,2,\dots, n)$. For $\epsilon \to 0$, it follows that $\frac{\opt(\t)}{\mech(\t,\pred)}=\frac{n}{1+(n-1)x} = \err $.
\end{proof}

\begin{lemma}\label{lem:ttc_lb_us}
    The TTC with recommended endowment mechanism is $\Omega(\min\{\err, n^2\})$-approximate for unit-sum valuations.
\end{lemma}
\begin{proof}
    Consider the instance in Figure \ref{fig:lb_ttc_with_endowment} (on the right) where the first agent has the following preferences $(\frac{1}{n}-\epsilon, \frac{1}{n}+\epsilon, \frac{1}{n}, \frac{1}{n}, \dots, \frac{1}{n})$. For the rest, agent $i$ has preferences $1-y$ for item $i$ and $y$ for item $(i+1)\mod n$, where $y = \frac{n(n-1)-\err+1}{n(n-1)(\err+1)}$. This value is smaller than $0.5$ because $y \leq \frac{1}{2} \iff \err \geq 1$ which clearly holds, and is also greater than 0 because $\err \leq n^2-n+1$ from the TTC mechanism lower bound \ref{lem:ttc_us}. Consider also the initial endowment $\pred = (2,3,4, \dots, n, 1)$. Again, the first agent will point to himself and agent $i$ will point to agent $(i-1)$. No exchanges will be made, resulting in a social welfare $\mech(\t,\pred) = \W(\t, \pred) = \frac{1}{n} + (n-1)y$. The optimal solution is again the main diagonal with a social welfare of $\opt(\t)=\frac{1}{n}-\epsilon + (n-1)(1-y)$. For $\epsilon \to 0$ and large values of $n$, it follows that the approximation ratio is $\frac{\opt(\t)}{\mech(\t,\pred)} = \frac{\frac{1}{n} + (n-1)(1-y)}{\frac{1}{n} + (n-1)y} = \err$.
\end{proof}

\begin{figure} [h]
    \centering
    \begin{subfigure}{0.35\textwidth}
        \centering
        \begin{tabular}{c|cccccc}
            $t_{ij}$&1&2&3&4&$\dots$&$n$\\
            \hline
            1&$1-\epsilon$& 1&0&0& $\dots$& 0\\
            2&0 & 1 & $x$&0&$\dots$& 0 \\
            3&0 & 0 & 1 &$x$&$\dots$&0\\
            4&0 & 0 & 0 &1 &$\dots$&0\\
            $\vdots$&$\vdots$& $\vdots$& $\vdots$&$\vdots$ & $\ddots$ & $\vdots$\\
            $n$&$x$& 0 & 0 & 0 & $\dots$ & 1 \\
        \end{tabular}
    \end{subfigure}%
    \begin{subfigure}{0.7\textwidth}
        \centering
        \begin{tabular}{c|cccccc}
            $t_{ij}$&1&2&3&4&$1\dots$&$n$\\
            \hline
            1&$\frac{1}{n}-\epsilon$& $\frac{1}{n}+\epsilon$&$\frac{1}{n}$&$\frac{1}{n}$& $\dots$& $\frac{1}{n}$\\
            2&0 & $1-y$ & $y$& 0&$\dots$& 0 \\
            3&0 & 0 & $1-y$& $y$ &$\dots$&0\\
            4&0 & 0 & 0 &$1-y$ &$\dots$&0\\
            $\vdots$&$\vdots$&$\vdots$& $\vdots$&$\vdots$ & $\ddots$ & $\vdots$\\
            $n$&$y$& 0 & 0& 0 & $\dots$ & $1-y$ \\
        \end{tabular}
    \end{subfigure}%
    \caption{TTC with initial endowment lower bound instances (unit-range valuations on the left \ref{lem:ttc_lb_ur} and unit-sum valuations on the right \ref{lem:ttc_lb_us}). After running TTC with initial endowment $a = (2,3,4, \dots, n, 1)$, the first agent gets item 2, whereas the rest get item $(i+1)\mod n$, $i>1$.}
    \label{fig:lb_ttc_with_endowment}
\end{figure}

A general approximation ratio for our mechanism immediately follows from the above.

\begin{theorem}\label{the:ttc_smoothness_tight}
    The TTC with recommended endowment mechanism is strategyproof, $\Theta(\min\{\err, n\})$-approximate for unit-range valuations and $\Theta(\min\{\err, n^2\})$-approximate for unit-sum valuations.
\end{theorem}

Using the characterization of \citet{svensson1999strategy}, we show that this performance is optimal among strategyproof, neutral and nonbossy mechanisms. A mechanism is neutral if its outcome is independent of the names of the indivisible goods, and nonbossy means that an individual cannot change the outcome of the mechanism without changing the outcome for himself at the same time.

\begin{theorem}\label{the:opt_ha}
    Every strategyproof, neutral and nonbossy mechanism in the house allocation problem with recommended matching is $\Omega(\min\{\err, n\})$-approximate for unit-range valuations and $\Omega(\min\{\err, n^2\})$-approximate for unit-sum valuations.
\end{theorem}

\begin{proof}
    According to \cite{svensson1999strategy}, if a house allocation mechanism is strategyproof, neutral and nonbossy, then it is a serial dictatorship. However, there is a 1-1 correspondence between serial dictator mechanisms and TTC mechanisms \cite{abdulkadirouglu1998random}. Therefore, our proof boils down to proving that the TTC mechanism with {\em any} initial endowment has the lower bound approximations in the theorem's statement. 
    
     Let $\pred$ be the recommended matching and suppose that a TTC mechanism considers an initial endowment $a'\neq \pred$. We rearrange the agents and the houses such that $a'=(2,3,4, \dots, n, 1)$ and $\pred_i\neq i$, for all $i$, and also $\pred_1\neq 2$; in Claim~\ref{cl:GlobalLBMatching} we show that such rearrangement always exists. We then construct instances for which the above TTC mechanism has approximation $\Omega(\min\{\err, n\})$ for range valuations. and $\Omega(\min\{\err, n^2\})$ for unit-sum valuations.
  
    \paragraph{Unit-range valuations} Consider the instance in Figure \ref{fig:optimal_ttc_ur}, where $z=\frac{(n/\err)-1}{n-1}$, for the case that $\err\leq n$. Agent 1 has preferences $1-\epsilon$ over house $1$, preference $1$ over house $a'_1=2$ and preference $1-\epsilon$ over house $\pred_1$ (which by construction is different from house $2$). Any other agent $i > 1$ has preference $1$ over house $i$, preference $z$ over house $a'_i=(i+1)\mod n$, and preference $z-\epsilon$ over house $\pred_i$ in case that $\pred_i \neq a'_i$. Note that the optimal solution is the main diagonal assignment with a social welfare $\opt = n-\epsilon$. As $\epsilon$ goes to $0$, $\W(\t,\pred)=z(n-1)+1=\frac{n}{\err}$, and the recommendation error becomes $\err$. After running the TTC algorithm with initial endowment $a'$, agents keep their initial house, as agent $i$ receives house $(i+1) \mod n$ on each round $i$. This gives $\W(\t, a') = z(n-1)+1=\frac{n}{\err}$, and so the mechanism has $\err$ approximation for this instance. 

    \begin{figure} [h]
    \centering
    \begin{tabular}{c|cccccc}
            $t_{ij}$&1&2&3&4&$\dots$&$n$\\
            \hline
            1&$1-\epsilon$& 1&0&$z-\epsilon$& $\dots$& 0\\
            2&0 & 1 & $z$&0&$\dots$& $z-\epsilon$ \\
            3&0 & 0 & 1 &$z$&$\dots$&0\\
            4&0 & $z-\epsilon$ & 0 &1 &$\dots$&0\\
            $\vdots$&$\vdots$& $\vdots$& $\vdots$&$\vdots$ & $\ddots$ & $\vdots$\\
            $n$&$z$& 0 & 0 & 0 & $\dots$ & 1 \\
        \end{tabular}
\caption{Optimality of TTC with recommended endowment mechanism, unit-range valuations case. Values in the diagonal represent the optimal solution. The TTC mechanism with initial endowment $a' = (2,3,4, \dots, n, 1)$ has asymptotically smaller approximation than the TTC mechanism with initial endowment $\pred$.}
    \label{fig:optimal_ttc_ur}
\end{figure}

    For the case that $\err > n$, we change the instance of Figure \ref{fig:optimal_ttc_ur} by setting $z=\frac{1}{\err}$ and the preference of agent $1$ for house $\pred_1$ to be also $\frac{1}{\err}$. The recommendation error is again $\err$ and $\opt = n-\epsilon$. The outcome of the TTC algorithm with initial endowment $a'$ is again $a'$ with $\W(\t, a')=1+\frac{n-1}{\err}<2$, so the mechanisms has $\Omega(n)$ approximation for this case. 
    
    \paragraph{Unit-sum valuations} Consider the instance in Figure \ref{fig:optimal_ttc_us}, where $y = \frac{n(n-1)+1-\err}{n(n-1)(\err+2)}$, for the case that $\err\leq n(n-1)+1$. Agent 1 has preferences $(\frac{1}{n}-\epsilon, \frac{1}{n}+\epsilon, \frac{1}{n}, \dots, \frac{1}{n})$ over the houses, while any other agent $i > 1$ has preference $1-2y$ over house $i$, preference $y$ over house $a'_i=(i+1)\mod n$ and preference $y-\epsilon$ over house $\pred_i$ if it is different from $a'_i$, otherwise, they have preference $y-\epsilon$ over some house that is different from $i$ and $a'_i$. $y$ is a decreasing function of $\err$ and therefore $y\leq \frac{1}{3}$ since always $\err \geq 1$. Therefore, the optimal solution is the main diagonal assignment with a social welfare $\opt = \frac{1}{n}+(n-1)(1-2y)$ (when $\epsilon$ goes to $0$). Moreover, as $\epsilon$ goes to $0$, $\W(\t,\pred)=\frac 1n + (n-1)y$, which after some calculations gives a recommendation error of $\err$. After running the TTC algorithm with initial endowment $a'$, agents keep their initial house, as agent $i$ receives house $(i+1) \mod n$ on each round $i$. This gives $\W(\t, a') = \frac 1n + (n-1)y$, and so the mechanism has $\err$ approximation for this instance. 
    
    For the case that $\err > n(n-1)+1$, we change the instance of Figure \ref{fig:optimal_ttc_ur} by setting $y=\frac{n(n-1)+1}{n^2(\err+2)-2}<\frac 1{n^2}$ and the preference of agent $1$ for house $\pred_1$ to be also $y$ and for the house $a'_i$ to be $\frac 2n +\epsilon - y$. It holds that the optimal allocation and the recommendation error remain the same. Since $y<\frac 1{n^2}$, it holds that $\opt = \frac{1}{n}+(n-1)(1-2y)>\frac n 2$ for $n\geq 3$. The outcome of the TTC algorithm with initial endowment $a'$ is again $a'$ with $\W(\t, a')=\frac 2n +(n-2) y<\frac 3n$, as $\epsilon$ goes to $0$, so the mechanisms has $\Omega(n^2)$ approximation for this case.

Thus, our mechanism is optimal among any TTC mechanism and consequently among all strategyproof, neutral and nonbossy mechanisms for the house allocation problem with recommended matching.

\begin{figure} [h]
        \centering
        \begin{tabular}{c|cccccc}
            $t_{ij}$&1&2&3&4&$\dots$&$n$\\
            \hline
            1&$\frac{1}{n}-\epsilon$& $\frac{1}{n}+\epsilon$&$\frac{1}{n}$&$\frac{1}{n}$& $\dots$& $\frac{1}{n}$\\
            2&0 & $1-2y$ & $y$& 0&$\dots$& $y-\epsilon$ \\
            3&0 & $y-\epsilon$ & $1-2y$& $y$ &$\dots$&0\\
            4&0 & $y-\epsilon$ & 0 &$1-2y$ &$\dots$&0\\
            $\vdots$&$\vdots$&$\vdots$& $\vdots$&$\vdots$ & $\ddots$ & $\vdots$\\
            $n$&$y$& 0 & 0& $y-\epsilon$ & $\dots$ & $1-2y$ \\
        \end{tabular}
    \caption{Optimality of TTC with recommended endowment mechanism, unit-sum valuations case. Values in the diagonal represent the optimal solution. The TTC mechanism with initial endowment $a' = (2,3,4, \dots, n, 1)$ has asymptotically smaller approximation than the TTC mechanism with initial endowment $\pred$.}
    \label{fig:optimal_ttc_us}
\end{figure}

\begin{claim}
\label{cl:GlobalLBMatching}
    Given two matchings $a'$ and $\pred$, with $a' \neq \pred$, there exists a rearrangement of the agents and the houses such that $a'=(2,3,4, \dots, n, 1)$ and $\pred_i\neq i$, for all $i$, and also $\pred_1\neq 2$.
\end{claim}

\begin{proof}
For simplicity of the presentation we will show that there exists a rearrangement of agents and houses where $a'=(1,2,3, \dots, n)$, $\pred_i\neq (i-1) \mod n$, for all $i$, and $\pred_1\neq 1$. Then, by renaming the houses so that house $i$ becomes house $(i+1) \mod n$, we get the desired rearrangement.

It is easy to see that by simply rearranging the houses we can have $a'_i=i$ for all agents $i$. We construct a graph with $n$ vertices, where each one represents an agent. We add a direct edge $(i,j)$ if and only if $\pred_j \neq i$. Each vertex has outdegree and indegree at least $n-2$. By Ghouila-Houri's Theorem \cite{GH60}, there exists a Hamiltonian cycle $C$. We rename the agents according to $C$ starting with an agent $i$ for which $a'_i\neq \pred_i$ (such an agent exists, since $a'\neq \pred$). Note now that under that renaming for each agent $i$, edge $((i-1)\mod n,i)$ exists in $C$, meaning that $\pred_i\neq (i-1) \mod n$, which concludes the proof of the claim.   
\end{proof}
\end{proof}

%% file: combAuctions.tex
In this section, we demonstrate how mechanism design with output
advice fit in a plug-and-play fashion with {\em maximal-in-range
  (MIR)} VCG mechanisms and more generally with affine maximizers,
when the social objective is to maximize an affine function. For
clarity, we focus on social welfare maximization, where VCG computes
the optimal outcome, but in many domains is intractable. The purpose
of MIR mechanisms is to narrow the range of the outcomes
$\mathcal{A}'\subseteq \mathcal{A}$ considered when maximizing the
objective, to facilitate computation, sacrificing optimality hence
resorting to approximation. Let $M$ be a maximal in range mechanism
with approximation guarantee $\rho_M$. We define the mechanism that
compares the output of mechanism $M$ with $\pred$ and keeps the
solution with the highest social welfare. This mechanism is still MIR
(as we simply replace $\mathcal{A}'$ with $\mathcal{A}'\cup \{\pred\}$)
 and therefore strategyproof. We formalize this idea through Lemma
\ref{lem:mir} and apply it to extract improved bounds in the context
of our model.

%algorithm
\begin{algorithm} [h]
\caption{MIR with recommended allocation}
\begin{algorithmic}[1]
\Require types $\t$, MIR mechanism $M$, recommendation $\pred$
\Ensure allocation $a$
\State $a_M \gets$ output of mechanism $M$
\State $a \gets \argmax\{\W(\t,a_M), \W(\t,\hat{a})\}$
\end{algorithmic}
\label{alg:mir}
\end{algorithm}

\begin{lemma}\label{lem:mir}
MIR with recommended allocation mechanism (Mechanism \ref{alg:mir}), denoted by $M$, is strategyproof and $min\{\err, \rho_M\}$-approximate.
\end{lemma}

\begin{proof}
Let $\mathcal{A}_M \subseteq \mathcal{A}$ be the restricted set of feasible allocations of mechanism $M$. The learning augmented mechanism finds the optimal allocation among $\mathcal{A}' = \mathcal{A}_M \cup \{\pred\}$. Hence, the mechanism is maximal in range and therefore strategyproof. Moreover:
$$\frac{\opt(\t)}{\mech(\t,\pred)} =\frac{\opt(\t)}{\max\{\W(\t,\pred), \W(\t,a_{M}\}} = \min\{\frac{\opt(\t)}{\W(\t,\pred)}, \frac{\opt(\t)}{W(\t,a_{M})}\}= \min\{\err, \rho_M\} \,.$$
\end{proof}

It immediately follows that the mechanism is 1-consistent in case the
recommendation is perfect, and $\rho_M$-robust in case the
recommendation is poor.

\subsection{Applications}
We now demonstrate the effectiveness of Lemma \ref{lem:mir} by
applying it to multi-unit auctions, combinatorial auctions with subadditive
valuations, and combinatorial auctions with
general valuations. In these settings, let $I$ be a set of $m$ indivisible
objects to be sold to $n$ bidders. An agent $i$ has private value
$t_i(B)$ for each bundle of items $B \subseteq 2^I$. Let $\b$ denote
the agents' bids. Items are allocated according to output allocation
$a = (a_1, \ldots, a_n), $ where $a_i$ corresponds to the bundle given
to agent $i$, hence $a_i\cap a_j=\emptyset$ for $i \neq j$. The
utility of each agent is the value for the bundle of the items they receive minus
the VCG payment, $u_i(\b, a) = t_i(a_i)-\p_i(\b)$. We wish to design
strategyproof mechanisms that maximize the social welfare $\W(\t,a)=\sum_i t_i(a_i)$.

\paragraph{Multi-Unit Auctions}
First, we apply Lemma~\ref{lem:mir} on multi-unit auctions with $m$
identical items and $n$ bidders. The number of items $m$ is viewed as
“large” and it is desired to achieve polynomial computational
complexity in the number of bits needed to represent $m$. Every bidder
$i$ has a value for obtaining $q$ items. We assume that $t_i$ is
weakly monotone and normalized. The best known mechanism achieves an
approximation ratio of
$2$~\cite{DBLP:journals/jair/DobzinskiN10}. This mechanism first
splits the items into $n^2$ bundles of size $\frac{m}{n^2}$ and
optimally allocates these whole bundles among the $n$
bidders. By using Lemma~\ref{lem:mir}, the mechanism that combines
$\pred$ with the above mechanism, is strategyproof, 1-consistent,
$2$-robust, and achieves a $\min\{\err, 2\}$ approximation ratio.

\paragraph{Combinatorial Auctions - General Valuations}
Next, we consider combinatorial auctions with general bidders. The
best known deterministic strategyproof mechanism obtains an
approximation ratio of
$\frac{m}{\log m}$~\cite{qiu2024settling}, which is optimal for strategyproof deterministic mechanisms \cite{qiu2024settling}. Following Lemma~\ref{lem:mir}, the mechanism
that combines $\pred$ with the MIR mechanism of \cite{qiu2024settling}, is strategyproof,
1-consistent, $\frac{m}{\log m}$-robust, and achieves a
$\min\{\err, \frac{m}{\log m}\}$ approximation ratio.

\paragraph{Combinatorial Auctions - Subadditive Valuations}

Finally, we apply Lemma~\ref{lem:mir} on the problem of combinatorial
auctions with polynomially many \textit{value} queries and
\textit{subadditive}\footnote{$t_i(S \cup T) \leq t_i(S) + t_i(T)$ for
  all bundles $S,T$} valuation functions. \citet{qiu2024settling} present a strategyproof mechanism that achieves a $\sqrt{\frac{m}{\log m}}$ approximation
ratio, which is also the best
approximation possible for strategyproof mechanisms with value queries~\cite{DBLP:journals/jacm/DobzinskiV16}. Following Lemma \ref{lem:mir}, this mechanism combined with the suggested outcome, is
strategyproof, 1-consistent, $\sqrt{\frac{m}{\log m}}$-robust, and achieves a
$\min\{\err, \sqrt{\frac{m}{\log m}}\}$ approximation ratio.

%% file: appendix.tex
\section{Missing Mechanisms from Section 3}

\subsection{The Bounding Box mechanism~\cite{agrawal2022learning}}\label{sec:MBB}
This mechanism runs the
MinMaxP procedure (Mechanism \ref{alg:MinMaxP}) on each one
of the two dimensions. Define the minimum bounding box surrounding the
agent locations on the Euclidean space.

\begin{algorithm} [h]
\caption{Minimum Bounding Box for egalitarian social cost in two dimensions~\cite{agrawal2022learning}}
\begin{algorithmic}[1]
\Require points $((x_{1}, y_{1}), \ldots,(x_{n}, y_{n})) \in \mathbb{R}^{2 n}$, prediction $(x_{\pred}, y_{\pred}) \in \mathbb{R}^{2}$
\Ensure $(x_{a}, y_{a})$
\State $x_{a}=\operatorname{MinMaxP}((x_{1}, \ldots, x_{n}), x_{\pred})$
\State $y_{a}=\operatorname{MinMaxP}((y_{1}, \ldots, y_{n}), y_{\pred})$
\end{algorithmic}
\label{alg:MinimumBoundingBox}
\end{algorithm}

{\floatname{algorithm}{Procedure}
\begin{algorithm} [h]
\caption{MinMaxP mechanism for egalitarian social cost in one dimension~\cite{agrawal2022learning}}
\begin{algorithmic}[1]
\Require points $(x_1, \ldots, x_n) \in \mathbb{R}^n$, prediction $\pred \in \mathbb{R}$
\Ensure $a$
\If{$\pred \in\left[\min _i x_i, \max _i x_i\right]$}
    \State $a \gets \pred$
\ElsIf{$\pred<\min_i x_i$}
    \State $a \gets \min_i x_i$
\Else
    \State $a \gets \max_i x_i$
\EndIf
\end{algorithmic}
\label{alg:MinMaxP}
\end{algorithm}}

\subsection{Coordinatewise Median with Predictions~\cite{agrawal2022learning}}\label{sec:CMP}

\begin{algorithm} [H]
\caption{Coordinatewise Median with Predictions in two dimensions~\cite{agrawal2022learning}}
\begin{algorithmic}[1]
\Require points $((x_{1}, y_{1}), \ldots,(x_{n}, y_{n})) \in \mathbb{R}^{2 n}$, prediction $(x_{\pred}, y_{\pred}) \in \mathbb{R}^{2}$, parameter\footnotemark $\lambda \in [0,1)$
\Ensure $(x_{a}, y_{a})$
\State $x_{a}=\operatorname{Median}((x_{1}, \ldots, x_{n}) \cup (x_{\pred}, \ldots, x_{\pred}))$ \Comment{$\left\lfloor \lambda n\right\rfloor$ copies of $x_{\pred}$}
\State $y_{a}=\operatorname{Median}((y_{1}, \ldots, y_{n})\cup (y_{\pred}, \ldots, y_{\pred}))$ \Comment{$\left\lfloor \lambda n\right\rfloor$ copies of $y_{\pred}$}
\end{algorithmic}
\label{alg:CMP}
\end{algorithm}
\footnotetext{$\lambda$ is a mechanism parameter that shows the confidence of the designer in the suggested outcome. The larger the $\lambda$, the more we trust the recommendation.}